\documentclass[10pt]{article}   	
\usepackage{amssymb, amsthm, amsmath, amsfonts, mathrsfs}
\usepackage{graphicx}
\graphicspath{{../Paper/images/}}
\usepackage{tikz}
\usepackage[all]{xy}
\usepackage[margin=1 cm, font=footnotesize]{caption}
\usepackage[margin=1 cm, font=footnotesize]{subcaption}
\usepackage[numbers,comma,square,sort&compress]{natbib}

\usepackage[letterpaper,top=.7in,left=.7in,right=.7in,bottom=.8in,
centering]{geometry}
\graphicspath{{images/}}

\newcommand{\ep}{\varepsilon}

\newcommand{\cc}{\mathbb{C}}
\newcommand{\rr}{\mathbb{R}}

\let\Im\relax
\DeclareMathOperator{\Im}{Im}
\DeclareMathOperator{\sgn}{sgn}

\DeclareMathOperator{\Id}{Id}
\newtheorem{thm}{Theorem}[section]
\newtheorem{lma}{Lemma}[section]

\newtheorem{rmk}{Remark}[section]

\usepackage{soul}
\soulregister\cite7
\soulregister\ref7
\soulregister\pageref7
\newcommand{\change}[1]{#1}
\newcommand{\typo}[1]{#1}
\newcommand{\dstypo}[1]{ #1}

\newcommand*\wc{{}\cdot{}}

\newcommand{\pt}[2]{\frac{\partial #1}{\partial #2}}

\newcommand{\ut}{\frac{\partial u }{\partial t}}

\makeatletter
\newcommand*{\defeq}{\mathrel{\rlap{%
                     \raisebox{0.3ex}{$\m@th\cdot$}}%
                     \raisebox{-0.3ex}{$\m@th\cdot$}}%
                     =}
\makeatother

\title{\bf Wavenumber selection via spatial parameter jump}
\author{ Arnd Scheel and Jasper Weinburd}
\date{ \small {\it University of Minnesota, School of Mathematics, 206 Church St. S.E., Minneapolis, MN 55455, USA} }

\begin{document}
\maketitle

\abstract{
The Swift-Hohenberg equation describes \change{an instability which forms finite-wavenumber patterns near onset}.
We study this equation posed with a spatial inhomogeneity; a jump-type parameter that renders the zero solution stable for $x<0$ and unstable for $x>0$. Using normal forms and spatial dynamics, we prove the existence of a family of steady-state solutions that represent a transition in space from a homogeneous state to a striped pattern state. The wavenumbers of these stripes are contained in a narrow band whose width grows linearly with the size of the jump. This represents a severe restriction from the usual constant-parameter case, where the allowed band grows with the square root of the parameter. We corroborate our predictions using numerical continuation and illustrate implications on stability of growing patterns in direct simulations. 
}

\section{Introduction}
In the classical thermal convection experiments of B\'enard, a \typo{shallow} plate of fluid is heated from below. For temperatures above a critical value, the fluid's diffusively heated state becomes unstable as heated fluid rises quickly in localized areas and cooler fluid falls nearby. \change{The rising and falling fluid of B\'enard's experiments created hexagonal convection cells, while subsequent convection experiments with similar settings produced squares and stripes (also called convection rolls). Each of these patterns occur with spatial periods in some range of a characteristic value \cite{ahlersinBook}}. Now suppose that we modify the experiment by heating only the right half-plate of fluid above the critical temperature. In this case, we may expect the fluid on the left to remain \typo{homogeneous} and the fluid on the right to form patterns. The analysis presented here confirms this intuition and \typo{additionally} determines that the range of spatial periods is significantly restricted from that occurring in the case where the full plate is heated. This selection of certain periods occurs in the full right-half plate, even far from the location of the temperature change.

\change{The pattern-forming convection experiments mentioned above are typically referred to as Rayleigh-B\'enard convection (RBC) and are a primary example of a \emph{finite-wavenumber instability}. Similar pattern-forming phenomena are widely observable throughout areas including biology \cite{kochMeinhardt}, optics and lasers \cite{longhi, legaNewell}, and chemical reaction-diffusion systems \cite{turing, castetsKepper, epstein}. 
In Turing's seminal paper on morphogenesis, he exhibits a finite-wavenumber instability using a simple two-species, reaction-diffusion model equation \cite{turing}. Later, in work which he never published, Turing wrote down a single-variable equation exhibiting his instability \cite{dawes}. Independently, and a quarter century later, Swift and Hohenberg used a very similar single-variable equation in their mathematical study of RBC \cite{swiftHohenberg}. Since then, the Swift-Hohenberg (SH) equation} 
\begin{align}\label{eqn:SH}
\ut = -\left(1 + \Delta
\right)^2 u +\mu u - u^3, \quad \quad u\in\rr, \quad \quad x\in \rr \text{ or } (x,y)\in\rr^2
\end{align}
\change{has been used as a model equation for a finite-wavenumber instability near onset. Today it is often used beyond its original convection context, for instance in modeling Turing patterns in embryonic development \cite{megason}, and appears extensively throughout the literature of pattern formation. It is known to possess stationary, even, spatially periodic solutions with wavenumber $k $ satisfying $|k^2-1|<\sqrt{\mu}$ \cite{colletEckmann, peletierTroy}. }

We consider a situation where the parameter $\mu$ varies over the spatial domain so that the homogeneous state $u\equiv 0$ is stable for $x<0$ and unstable for $x>0,$ 
\begin{align} \label{eqn:SHjump}
\ut = -\left(1 + \frac{\partial^2}{\partial x^2}\right)^2 u +m(x) u - u^3, \quad \quad u,x\in \rr, \quad \quad m(x)=\begin{cases}\mu, & x>0\\ -\mu, & x<0\end{cases}.
\end{align} 
For small $\mu>0$, we prove the existence of half-patterned stationary solutions to \eqref{eqn:SHjump} and show that their wavenumbers occur in an interval significantly narrower than the wavenumbers of solutions to \eqref{eqn:SH}; see Figures \ref{fig:boundsonk} and \ref{fig:samplesolns}.


\begin{figure}[h]
 \centering\includegraphics[scale=0.4]{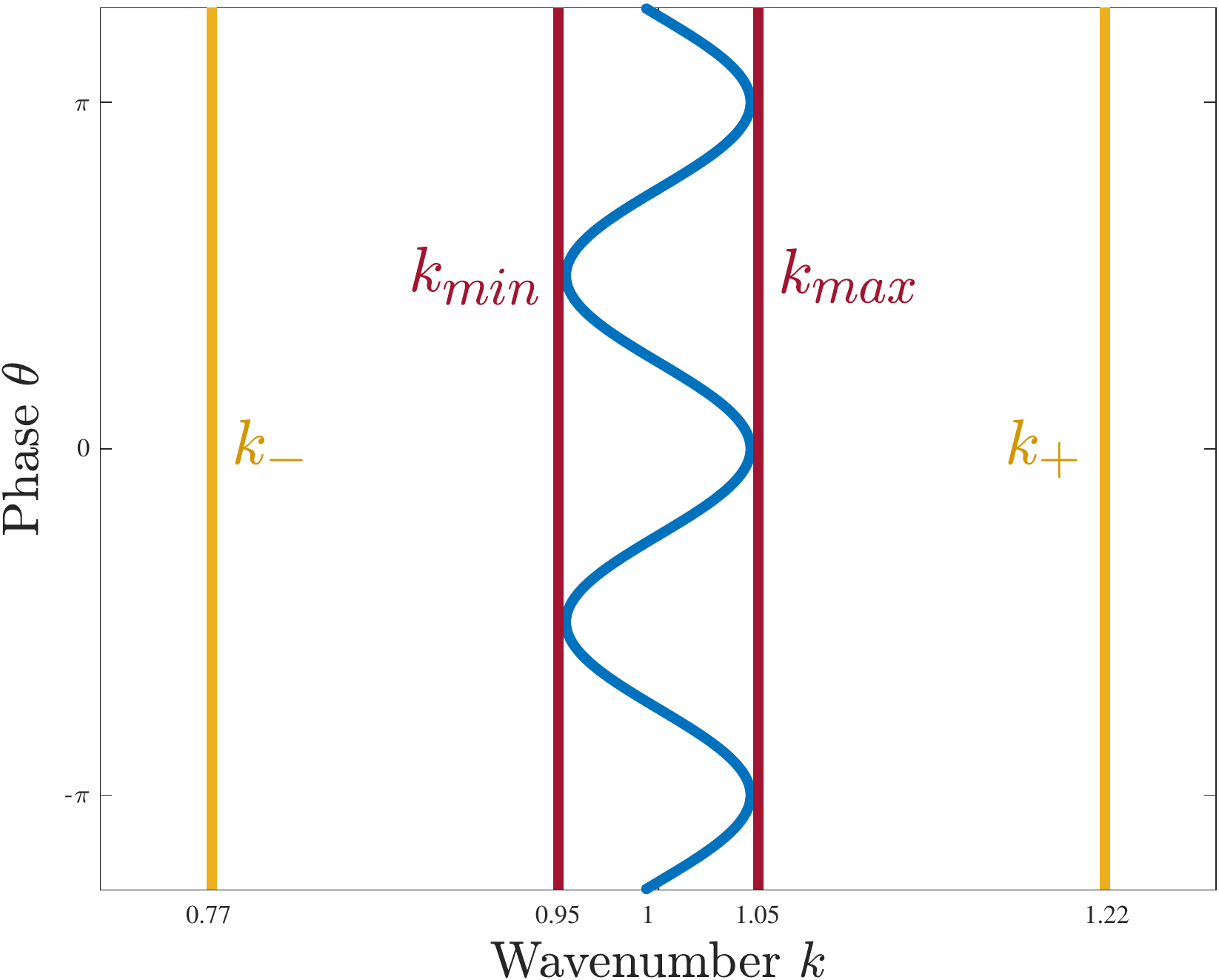}\quad
 \raisebox{-0.045in}{\includegraphics[scale=0.4]{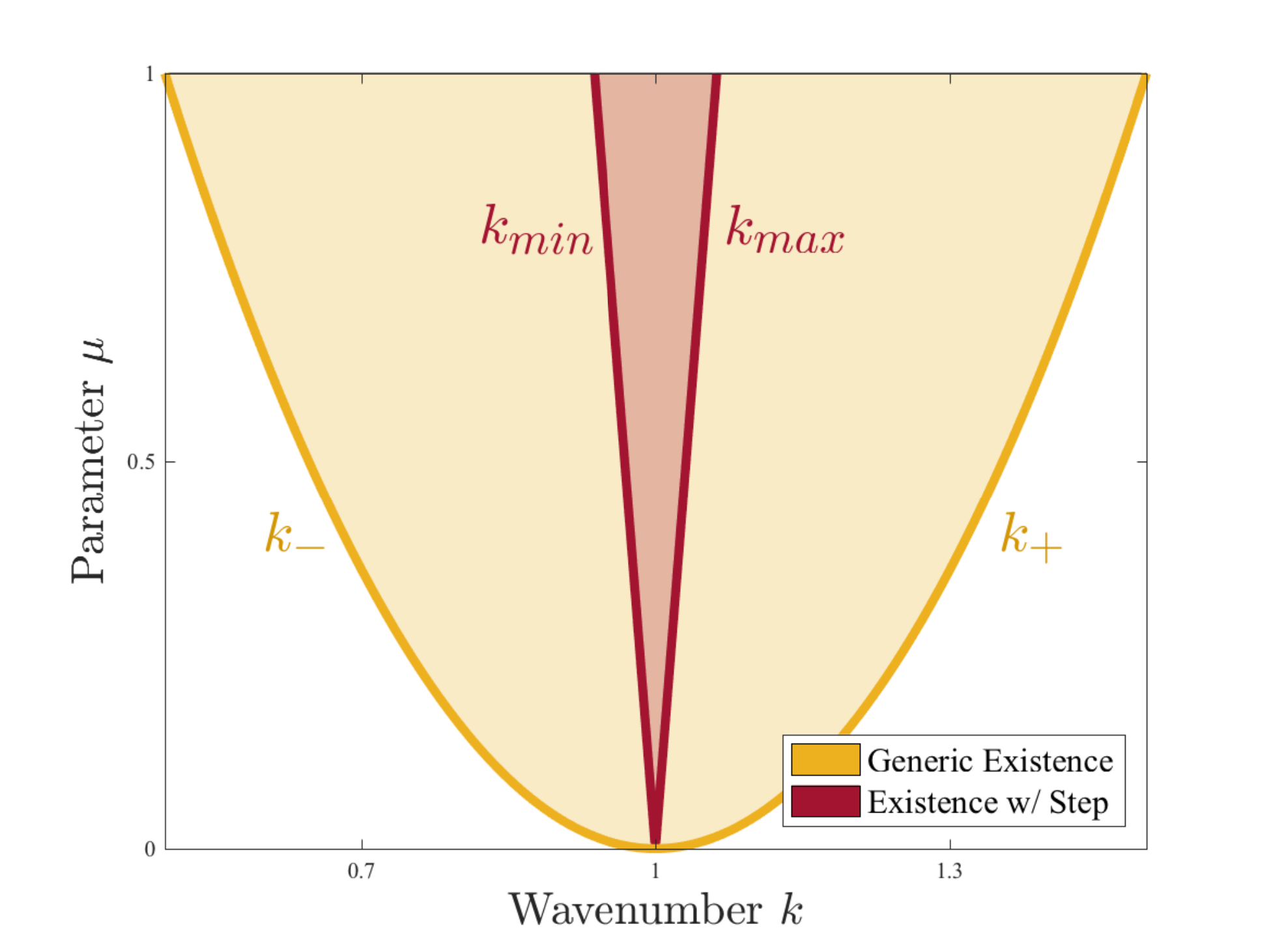}}
 \caption{Bounds on wavenumbers $k$ of solutions to the Swift-Hohenberg equation with constant parameter ($k_\pm$, gold) and with jump-type parameter ($k_{\text{min/max}}$, red). Left: With $\mu = 0.8$, strain-displacement relation (blue) and bounds on wavenumber (gold, red). Right: Existence region of stripe solutions to \eqref{eqn:SH} (gold) and half-stripe solutions to \eqref{eqn:SHjump} (red).}
   \label{fig:boundsonk}
\end{figure}

\change{Pattern formation} in inhomogeneous environments has been investigated experimentally. In the thermal convection setting, experiments were performed as early as the 1980s \cite{walton}. More recent experiments include the effect of applying a local electric field to a nematic liquid crystal, producing zigzag patterns \cite{Andrade-Silva}. Harkening back to Turing's original paper, biologists have used heterogeneous parameters in modeling the production of morphogens in embryonic development \cite{umulis, kaper, megason}. Additionally, there is interest in exploiting a parameter jump for circuit construction for storage and processing in classical and quantum settings \cite{pegrum, goldobin}.

\change{Furthermore, jump-type inhomogeneities in the context of pattern formation have been investigated in the mathematical literature, for instance \cite{benson, derks, jonesNLS, doelman, gohScheel, gohBeekie, jaramillo, scheelMorrissey}. In Turing's reaction-diffusion model, the idea was formally investigated using asymptotics in \cite{benson}. 
The analysis in \cite{doelman} for the nonlinear wave and Schrödinger equations is global, but relies on rather explicit knowledge of the phase portraits in spatial dynamics. In contrast, the analysis of reaction-diffusion spikes in \cite{doelman} is global but perturbative in nature.
Previous work by Scheel, Goh, and others investigates existence of non-stationary striped wave-trains in the case of a moving parameter jump \cite{gohScheel} and also in slowly-growing domains \cite{gohBeekie}, which may be seen as an analogue of a slowly moving jump. Effects of localized impurities are studied in \cite{jaramillo}. In \cite{scheelMorrissey}, Morrissey and Scheel develop basic concepts useful for stating our main result, which we discuss in the next section.
Many of our techniques, such as overlapping phase portraits, were previously used in these scenarios. However, our use of normal form theory appears to be unique in the study of patterns in inhomogeneous environments.}

The most closely related work was produced in the 1980s, in the wake of renewed interest in RBC experiments. Groups centered around M. Cross and L. 
Kramer, and Y. Pomeau explored situations with a parameter that varies through criticality on a slow spatial scale \cite{CrossKramer, pomeauZaleskiRamps}. The same groups also explored the semi-finite system with boundary conditions \cite{CDHS, pomeauZaleskiBC}, which bear more heuristic similarities to our parameter jump. For further discussion of \change{boundary conditions as analogous to} the present parameter inhomogeneity, see Section \ref{sect:discussion}.

\subsubsection*{Strain-Displacement Relations} \label{sect:SDRel}

Our main theorem (below) may be thought of as the computation of a strain-displacement relation. A strain-displacement relation describes the phase and wavenumber of a certain pattern that may occur in a given system. For a full description of strain-displacement relations see \cite{scheelMorrissey}; we provide only a brief explanation here.

Consider the stationary Swift-Hohenberg equation
\[
0 = -\left(1 + \frac{\partial^2}{\partial x^2}\right)^2 u +m(x) u - u^3.
\]
For $m(x)\equiv \mu>0$, there exists a family of stationary, even, spatially periodic solutions, that we shall refer to as {\it stripes} \cite{colletEckmann, peletierTroy}. Parameterizing the family of stripes by phase $\theta$ and wavenumber $k$, we write $\mathcal S = \{ u_*(k x - \theta; k) \mid \theta\in [0,2\pi), k\in J_k\}$ for some open interval $J_k$. Writing the equation as a system of first-order equations in $\rr^4$, $\mathcal{S}$ \typo{corresponds to} a family of periodic orbits, forming a smooth annulus. It turns out that part of this annulus is normally hyperbolic, that is, each periodic orbit possesses 2-dimensional stable and unstable manifolds. For $m(x)\equiv -\mu<0$, the origin is hyperbolic, with two-dimensional stable and unstable manifolds. One obtains solutions for the case of our jump function $m(x)=\mu \sgn(x)$ by matching solutions to $m(x)=\mu$ and $m(x)=-\mu$ at $x=0$, overlaying the ``phase portraits'' from the two cases. 
The solutions of interest to us are solutions that converge to zero in $x<0$ and to a periodic solution in $x>0$. Those solutions are found in the intersection of the 2-dimensional unstable manifold of the origin from the  $m(x)\equiv -\mu$ flow, and the 3-dimensional stable manifold of the family of periodic orbits from the $m(x)\equiv\mu$ flow. Adding dimensions of manifolds and subtracting the dimension of ambient space, $2+3-4=1$, we expect a one-dimensional curve of intersections. Points on the curve can be identified with the asymptotic periodic orbit's wavenumber $k$ and its asymptotic phase $\theta$. We refer to the relation between $k$ and $\theta$ on this curve as the \emph{strain-displacement relation}, alluding to the intuitive stretching and compression, due to variations in $k$, and the displacement or shift of the asymptotic pattern, due to changes in $\theta$. 

In \cite{scheelMorrissey}, strain-displacement relations refer to the situation of the SH equation posed on $x>0$, with boundary conditions at $x=0$. Those boundary conditions can be viewed as a two-dimensional manifold in the associated 4-dimensional ODE, equivalent to the unstable manifold of the origin in the $m(x)\equiv -\mu$ system. From this view point, the parameter jump quite literally represents an effective boundary condition at $x=0$ (see Section \ref{sect:discussion}). 

Figure \ref{fig:boundsonk} shows a strain-displacement relation (blue) in the $(k,\theta)$ plane. Figure \ref{fig:samplesolns} shows two sample solutions with phase-dependent wavenumbers. The top sample solution corresponds to the point $(k,\theta) = (1.05,0)$ on the strain-displacement relation of Figure \ref{fig:boundsonk} and the bottom solution corresponds to the point $(k,\theta) = (0.95,\tfrac{\pi}{2})$. The difference in asymptotic phases may be observed at $x=0$ and the difference in wavenumbers may be observed by counting minima.

\begin{figure}[h]
\centering
\includegraphics[width=0.38\textwidth]{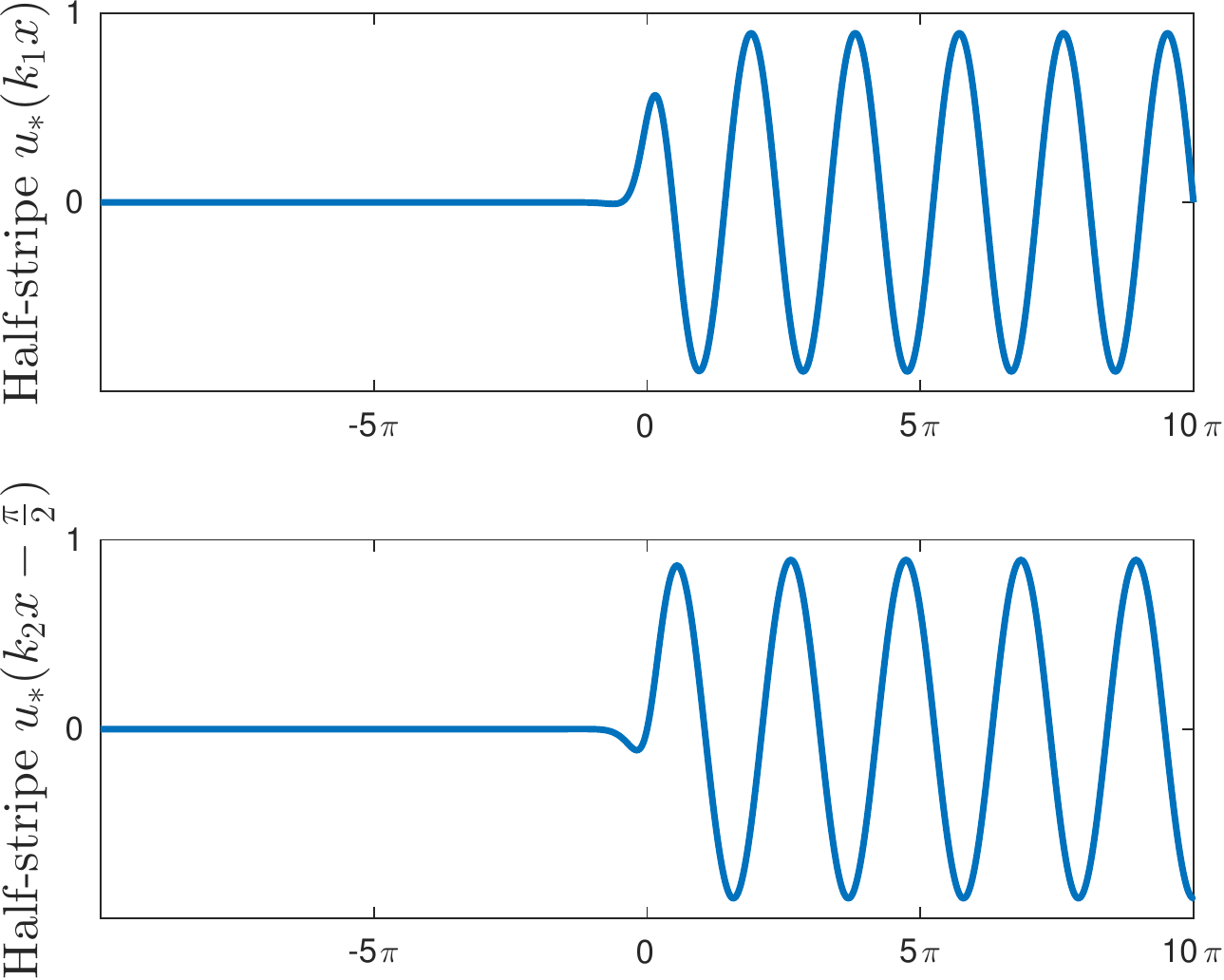}
\qquad \qquad
\raisebox{0.42in}{\includegraphics[width=0.4\textwidth]{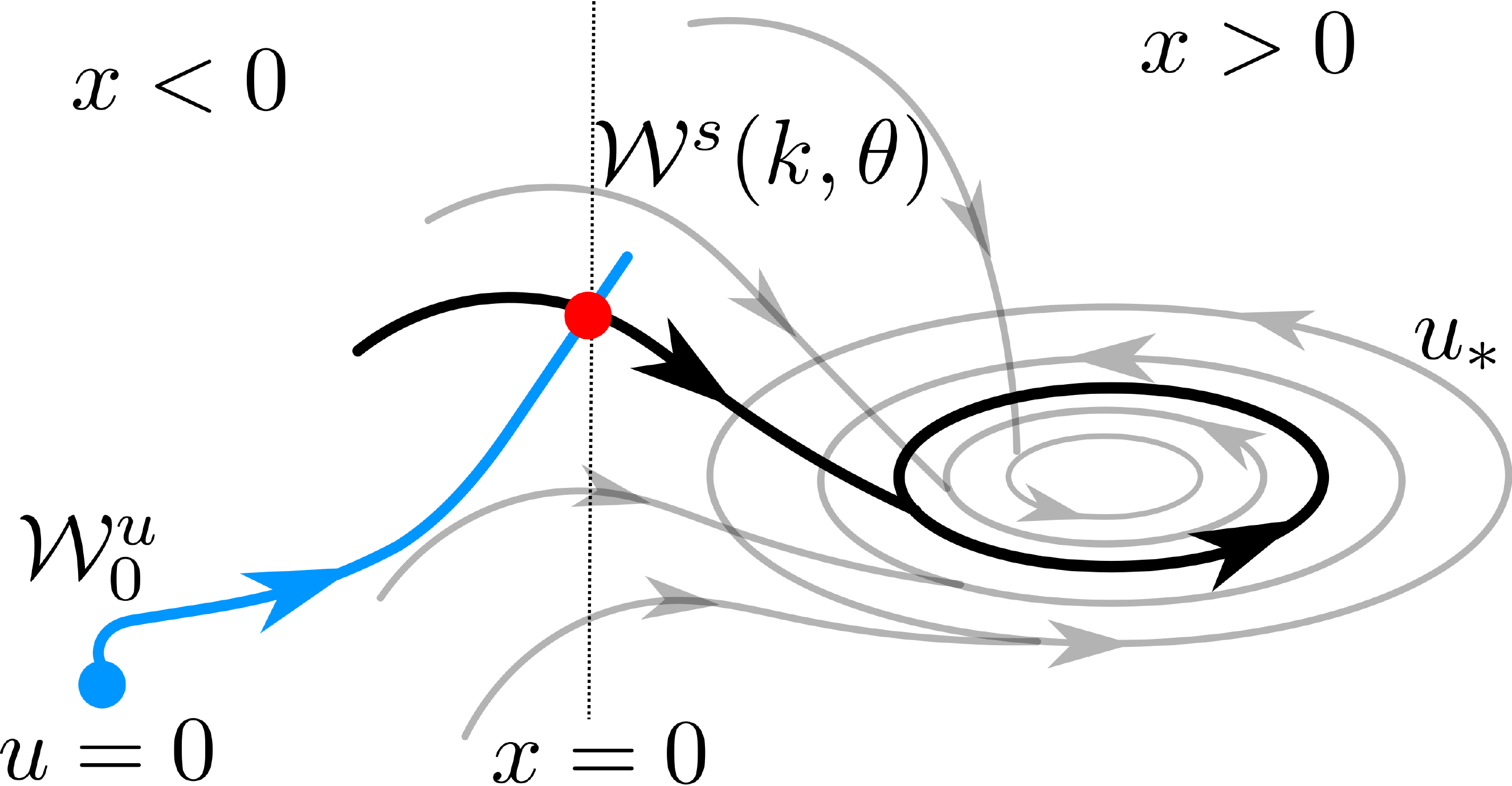}}
\caption{Left: Half-stripe solutions with asymptotic phase shifts $0$ (top) and $\tfrac{\pi}{2}$ (bottom). The dependence of wavenumber on phase is illustrated by \typo{the difference in} number of minima. Right: Schematic showing spatial dynamics with a heteroclinic from $u\equiv 0$ to a periodic solution.}\label{fig:samplesolns}
\end{figure}

\subsubsection*{Main Result and Outline}

Our main result is the existence of stationary \emph{half-stripe} solutions to the SH equation with a parameter jump, equation \eqref{eqn:SHjump}.
These half-stripe solutions tend to $0$ as $x \to -\infty$ and converge to a spatially periodic solution as $x\to \infty$. Furthermore, we determine the leading order terms in the expansion for the strain-displacement relation. We find that the system selects wavenumbers in a band growing linearly in the size of the parameter jump $\mu$.

\begin{thm} \label{thm:mainthm}
For sufficiently small $\mu>0$, there exists a one-parameter family of bounded, stationary solutions of the Swift-Hohenberg equation \eqref{eqn:SHjump}. The family $u(x;\theta) = u(x;\theta + 2\pi)$ is periodic in the parameter $\theta$. For each solution $u(x;\theta)$ in the family, the following hold:
\begin{enumerate}
\item $u(x;\theta) \to 0$ as $x\to -\infty$, and
\item $|u(x;\theta) - u_*(kx -\theta;k )| \to 0$ as $x\to \infty$,
\end{enumerate}
for some $k=k(\theta;\mu)$, where $u_*(kx ;k )$ is an even periodic solution to \eqref{eqn:SH} with maximum at $x=0$ and wavenumber $k$ depending on the parameter $\mu$.
Furthermore, at leading order, 
$$k 
(\theta; \mu) = 1 + \frac{\mu}{16}\cos 2\theta + \mathcal O (\mu^{3/2})
$$
with $\theta$-uniform higher-order corrections.
\end{thm}

\change{This result establishes the shape of a strain-displacement relation to leading order and at small amplitude. We are not aware of other cases where such a characterization has been carried out, except for the cases studied in \cite{scheelMorrissey} where strain-displacement relations could be computed explicitly in the Ginzburg-Landau equation. Additionally, our result provides a rigorously established example where results on wavenumber selection in growing domains, from \cite{gohBeekie}, are applicable. In this way, our result broadens the spectrum covered by the theories in \cite{scheelMorrissey,gohBeekie}, adding a novel and universal example.}

We give our proof in Section \ref{sect:proof}. The proof relies on normal form theory, allowing us to use the real Ginzburg-Landau equation as amplitude equation for the SH equation. We explicitly compute the $\mu$-dependent part of the normal form transformation in order to bridge the parameter jump. Having overcome this difficulty, we construct a family of heteroclinic connections in the Ginzburg-Landau equation which correspond to the half-stripe solutions.

In Section \ref{sect:application} we present numerical corroboration and an application via direct simulation. We use numerical continuation to compute strain-displacement relations which have good agreement with Theorem \ref{thm:mainthm}. We also apply a two-dimensional analogue of Theorem \ref{thm:mainthm} in order to select a zigzag pattern. To do so, we discuss the two-dimensional instabilities of stripes in the SH equation posed on the plane, explain how slowly moving the jump allows us to select a unique wavenumber from the restricted band, and briefly describe the resulting zigzag patterns.

In Section \ref{sect:discussion} we discuss some implications and extensions of our result including the use of actual boundary conditions, smooth parameter inhomogeneities (including very slow ramps), stability, and extensions to $(x,y)\in\rr^2$.
\paragraph{Acknowledgments.} The authors gratefully acknowledge supprt through NSF DMS-1612441.

\section{Proof: Existence of Half-Stripes and Wavenumber Selection} \label{sect:proof}
We begin with some preliminary lemmas, which are all stated for the constant parameter problem \eqref{eqn:SH}. In Section \ref{sect:diff} we justify how we apply these results for $x<0$ and $x>0$ independently, then overlay phase portraits.

\subsection{Spatial Dynamics}
We begin by writing the steady-state equation with a constant parameter $\mu$ as a first order equation in $\rr^4$. So
\[
0  =  u + 2u_{xx} +u_{xxxx} - \mu u + u^3
\]
becomes
\begin{align} \label{eqn:U}
 \frac{dU}{dx}  =   L U + R(U;\mu),
\end{align} 
where
\[
U = \begin{bmatrix}
  u \\
  u_x \\
  u_{xx}+u\\
  u_{xxx}+u_x
 \end{bmatrix} \in \rr^4, \quad
L  = \begin{bmatrix}
  0 & 1 & 0 & 0 \\
  -1 & 0 & 1 & 0 \\
  0 & 0 & 0 & 1 \\
  0 & 0 & -1 & 0
 \end{bmatrix}, \quad 
 R(U,\mu) = \begin{bmatrix}
  0 \\
  0 \\
  0 \\
  \mu U_1
 \end{bmatrix} + \begin{bmatrix}
  0 \\
  0 \\
  0 \\
  -U_1^3
  \end{bmatrix}.
 \]
  
In this form, our equation undergoes a ``reversible-Hopf bifurcation" or a ``reversible $1:1$ resonance" as described in \cite[\S4.3.3]{ioossHaragus} at $\mu=0$. The symmetry $u(x) = u(-x)$ of the SH equation has been replaced by the reversibility symmetry $\mathbf S$ defined by $\mathbf S U = (U_1,-U_2,U_3, -U_4)^\top$, 
 which anti-commutes with the vector field on the right-hand side of Equation \eqref{eqn:U}. This reversibility plays an important role in the computation of the normal form equation and transformations below.
  
\subsection{Normal Forms} \label{subsect:NF}
In this section, we present a normal form equation which represents the dynamics of the ODE in the last section for sufficiently small $\mu$ values. Our normal form equation is the same as that found in \cite{ioossHaragus}. However, we also compute the leading $\mu$-dependent part of the transformation used to arrive at the normal form equation; it is essential for our future steps.

We define the space $\widetilde{\cc^2} \defeq \cc^4/\langle(A,B,C,D)-(\overline C,\overline D,\overline A, \overline B)\rangle = \{(A,B,\overline A, \overline B)\mid A,B\in \cc\}$. Clearly $\widetilde{\cc^2}\cong \cc^2$, and so we drop the $\widetilde{\phantom{nn}}$ to simplify notation. Note that the matrix
$$ \Theta = \begin{bmatrix} 
1 & 0 & 1 & 0 \\
i & 1 & -i & 1\\
0 & 2i & 0 & -2i \\
0 & -2 & 0 & -2 \\
\end{bmatrix} \quad \quad \text{maps}\quad\quad \cc^2 \to \rr^4.$$
The change of variables $U =\Theta (A,B)$ puts the linear part $L$ into Jordan normal form and yields an equation of the form
\begin{align} \label{eqn:jordanform}
\begin{bmatrix} A_x\\ B_x \end{bmatrix} & = \begin{bmatrix} i & 1 \\0 & i \end{bmatrix} \begin{bmatrix} A\\ B \end{bmatrix} + \widetilde R(A,B;\mu)
\end{align}
and the complex conjugate equations.


\begin{lma}[Normal Form Equation] \cite[Lem. 3.17]{ioossHaragus} \label{lma:normalform}
For any positive integer $N\geq 1$, there exist neighborhoods $\mathcal{U},\mathcal{V}$ of $0$ in $\cc^2$ and $\rr$ respectively so that for any $\mu\in \mathcal{V}$ there exists a polynomial $\Phi(\wc;\mu)\colon\cc^2 \to \cc^2 $ of degree $N$ with the following properties:
\begin{enumerate}
\item The coefficients of the monomials of degree $q$ in $\Phi(\wc;\mu)$ are functions of $\mu$ of class $\mathcal{C}^{N-q}$,
$$ \Phi(0,0;0) = 0, \quad \quad \partial_{(A,B,\overline A,\overline B)} \Phi(0,0;0) =  0, \quad \text{and} \quad \mathbf{S}\Phi(A,B;\mu) = \Phi(\overline A,-\overline B;\mu). $$
\item For $(A,B)\in \mathcal{U}$, the change of variables
\begin{align}\label{eqn:nftransfn}
(A,B) \mapsto  \Id + \Phi(A,B;\mu)
\end{align}
gives a transformation $\cc^2 \leftrightarrow \cc^2$ which transforms equation \eqref{eqn:jordanform} into
\begin{align} \label{eqn:thenormalform}
\begin{bmatrix} A_x\\ B_x \end{bmatrix} & = \begin{bmatrix} i & 1 \\0 & i \end{bmatrix} \begin{bmatrix} A\\ B \end{bmatrix} + \begin{bmatrix} iP(A,B;\mu) & 0 \\Q(A,B;\mu) & iP(A,B;\mu)\end{bmatrix} \begin{bmatrix} A\\ B \end{bmatrix} + G(A,B;\mu)
\end{align}
where the remainder $G$ is smooth and $G(A,B;\mu) = o((|A|+|B|)^N)$ and $P,Q$ are real valued polynomials of degree $N-1$ given by
\begin{align*}
P(|A|^2, (A\overline B-\overline A B);\mu) & = -\tfrac{1}{8}\mu + \tfrac{9}{16}|A|^2 + \mathcal{O}\left((|\mu|+(|A|+|B|)^2)^2\right)\\
Q(|A|^2, (A\overline B-\overline A B);\mu)& = -\tfrac{1}{4}\mu + \tfrac{3}{4} |A|^2 + \tfrac{3i}{16}(A\overline B-\overline A B) + \mathcal{O}\left((|\mu|+(|A|+|B|)^2)^2\right).
\end{align*}
\end{enumerate}
\end{lma}

This lemma is a restatement of Lemma 3.17 \cite[\S4.3.3]{ioossHaragus} in the particular case of a double eigenvalue equal to $i$. In an example in the same section
, the authors compute the first three coefficients in each of the polynomials $P,Q$ as they appear above. To do this they execute part of an algorithmic computation which is derived from their proof of the normal form theorem. The same algorithm may be used to compute the normal form transformation itself $\Id + \Phi(A,B;\mu)$. However, the authors stop short of this detail. We must compute part of it explicitly for a later argument.

\begin{rmk}
In practice, we compute this transformation as a composition of functions, each accurate up to a certain order in $\mu$ and $|A|, |B|$. \change{We use the notation $\Phi_{p,q}$ for a polynomial with degree $p$ in $A$ and $B$, and degree $q$ in $\mu$. In our case, we first compute the cubic  (in $|A|,|B|$) polynomial $\Phi_{3,0}$ with $\mu = 0$.
Then we compute the lowest order $\mu$-dependent polynomial $\Phi_{1,1}$ (which is linear in both $\mu$ and $|A|,|B|$).} 
Now the transformation of the lemma is precisely the composition $\Id + \Phi(A,B; \mu) = (\Id + \Phi_{3,0})\circ( \Id + \Phi_{1,1})$.
For more detail on how the normal form transformations at different orders depend on each other, see \cite[\S3.2.3]{ioossHaragus}.
\end{rmk}

Using the algorithm from \cite{ioossHaragus} we compute the vector coefficients in the polynomial for $N=1$. The computations essentially amount to applications of the Fredholm alternative and solving systems of four linear equations. 

\begin{lma} [Normal Form Transformation] \label{lma:nftransfn}
Let $\mathcal{V}$ be the neighborhood guaranteed by Lemma \ref{lma:normalform}. For any $\mu\in \mathcal{V}$, a polynomial satisfying the conditions of Lemma \ref{lma:normalform} with $N=1$ is
\begin{align} \label{eqn:mutrans}
\Phi_{1,1}(A,B;\mu) & = \mu \begin{bmatrix} 
\frac{-3}{16} & \frac{-i}{8} & \frac{3}{16}  & \frac{-i}{8} \\[6pt]
\frac{i}{8}& \frac{-3}{16} & \frac{-i}{8} & \frac{-1}{16}  \\[6pt]
\frac{3}{16}  & \frac{i}{8} & \frac{-3}{16}  & \frac{i}{8} \\[6pt]
\frac{i}{8} & \frac{-1}{16} & \frac{-i}{8} & \frac{-3}{16} \\
\end{bmatrix} \begin{bmatrix} A\\[6pt] B\\[6pt] \overline A\\[6pt] \overline B\end{bmatrix}
\end{align}
\end{lma}

\begin{rmk}
Lemma \ref{lma:nftransfn} describes a smooth unfolding of the $\mu$-dependent linear terms in the equation \eqref{eqn:jordanform}. This contrasts with a classical Jordan normal form transformation of these terms which is not smooth in $\mu$. In this context, the lemma may be compared to Section 3.2.2 of \cite{ioossHaragus} or to the original source \cite{arnold}.
\end{rmk}

There are a few important features of the normal form equation \eqref{eqn:thenormalform}. When truncated by removing $G$, it is equivariant under the reversibility symmetry $\mathbf{S}$ and a Gauge symmetry. Furthermore, it possesses a pair of conserved quantities. We will introduce these after further transformations.

\subsection{Additional Transformations}
To simplify notation, we set $\ep =\sqrt{|\mu|}$.
\begin{lma} \label{lma:othertrans}
There exists a change of coordinates which transforms Equation \eqref{eqn:thenormalform} into
\begin{align} \label{eqn:othertrans}
\begin{bmatrix} a_y\\ b_y \end{bmatrix} & = \begin{bmatrix} 0 & 1 \\0 & 0 \end{bmatrix} \begin{bmatrix} a\\ b \end{bmatrix} + \begin{bmatrix} i\ep \widetilde P(a,b;\ep) & 0 \\ \widetilde Q(a,b;\ep) & i\ep \widetilde P(a,b;\ep)\end{bmatrix} \begin{bmatrix} a\\ b \end{bmatrix} + \ep^{N-2}\widetilde G(a,b,y/\ep;\ep)
\end{align}
where the transformed terms in normal form are
\begin{align*}
\widetilde{P}(a,b,\ep) & = -\sgn(\mu)\frac{\dstypo{1}}{4}+ \frac{\dstypo{3}}{8}|a|^2 +\mathcal{O}\left(\ep^2\right)\\
\widetilde{Q}(a,b,\ep) &=  -\sgn(\mu) + |a|^2 + \ep \frac{i}{8}(a\overline{b}-\overline{a}b) + \mathcal{O}(\ep^3)
\end{align*}
and the remainder term $\widetilde{G} = 
\mathcal{O}((|a|+\ep|b|)^N)$ is periodic in the third variable.
\end{lma}

\begin{proof}
The transformation may be realized as a composition of two transformations. First, we move to a co-rotating frame of reference. Second, we rescale $a, b,$ and space $x$.

To move to the co-rotating frame, we essentially apply the change of coordinates
\begin{align*}
\nu\colon \cc^2\times S^1 & \to \cc^2\\
(A_1,B_1, e^{i x}) & \mapsto (A_1e^{ix}, B_1e^{ix}) = (A,B).
\end{align*}
We may formalize this by introducing a rotation variable $\rho = e^{ix}$ and appending the equation $\rho_x  = i\rho$, then applying the transformation. We omit the details for brevity.

The rescaling transformation takes the form
\begin{align*}
\tau\colon \cc^2 \times \rr & \to \cc^2\times \rr\\
(a,b,y) & \mapsto \left(\tfrac{\ep}{\sqrt{3}} a,\tfrac{\ep^2}{2\sqrt{3}}b,\tfrac{2}{\ep}y\right) = (A_1,B_1,x)
\end{align*}
Equation \eqref{eqn:othertrans} follows after some algebra.
\end{proof}

In order to examine the dynamics of equation \eqref{eqn:othertrans}, we may examine the truncated system with $\widetilde G$ removed. To justify this, we must argue for the persistence in the full system of \typo{certain invariant manifolds} which we will find in the truncated system. First, by using a higher degree normal form, we make $N$ arbitrarily large and can thus control the size of these terms (since $\ep$ is small). However, the persistence we desire is not guaranteed by standard invariant manifold theory since $\widetilde G$ has a period which tends to $0$ as $\ep \to 0$. Thinking of $\widetilde G$ as a rapidly oscillating perturbation, we may imagine that its effects are averaged out over any given period (at least when $\ep$ is small enough). Indeed, this intuition holds. For a rigorous discussion of the persistence of the relevant manifolds, see \cite{ioossPeroueme}. With this consideration behind us, we may examine the truncated equation
\begin{align} \label{eqn:truncatedNF}
\begin{bmatrix} a_y\\ b_y \end{bmatrix} & = \begin{bmatrix} 0 & 1 \\0 & 0 \end{bmatrix} \begin{bmatrix} a\\ b \end{bmatrix} + \begin{bmatrix} i\ep \widetilde P(a,b;\ep) & 0 \\ \widetilde Q(a,b;\ep) & i\ep \widetilde P(a,b;\ep)\end{bmatrix} \begin{bmatrix} a\\ b \end{bmatrix} 
\end{align}
 
Now we turn to our problem of interest in equation \eqref{eqn:SHjump}, with a spatial inhomogeneity in the parameter $m(x) = \mu \sgn (x)$, for small $\mu>0$. We apply the preceding lemmas for $x<0$ and $x>0$ separately, obtaining two different equations in place of \eqref{eqn:truncatedNF}. After dropping the vector notation and plugging in $\widetilde P, \widetilde Q$, these become 
\begin{align} \label{eqn:firstord1}
\frac{da}{dy} & = b + \ep \left(\dstypo{-\sgn(y)}\frac{i}{4} a+ \frac{3i}{8}a |a|^2\right) + \mathcal{O}(\ep^2|a|)\\
\label{eqn:firstord2}
\frac{db}{dy} & =  -\sgn(y)a + a|a|^2+ \ep\left( \dstypo{-\sgn(y)}\frac{i}{4} b+ \frac{3i}{8} b |a|^2 + \frac{i}{8}a(a\overline b-\overline a b) \right)+ \mathcal{O}(\ep^2|a|)+ \mathcal{O}(\ep^3|b|)
\end{align}
where we have denoted two pairs of equations by the use of $\sgn (y)$ (recall that we have rescaled space $x = \tfrac{2}{\ep}y$).
For each choice of $y>0$ or $y<0$, equations \eqref{eqn:firstord1}--\eqref{eqn:firstord2} are equivariant under the reversibility symmetry $\mathbf S$ and also the Gauge symmetry $\mathbf R_\phi \colon (a,b) \mapsto (e^{i\phi}a,e^{i\phi}b)$. Additionally, they have the conserved quantities
\begin{align} \label{eqn:1stintl}
M & = \dstypo{-}\frac{i}{2}(a\overline b-\overline a b) = \Im(a\overline b), \quad\quad\quad
H^\pm(\ep) = |b|^2 \pm|a|^2-\tfrac{1}{2}|a|^4 + \ep \tfrac{1}{4}|a|^2M + \mathcal{O}(\ep^2). 
\end{align}
These conserved quantities may be computed by applying $\nu\circ \tau$ to the conserved quantities of the normal form, found in \cite[\S4.3.3]{ioossHaragus}.

\subsection{Different Phase Spaces} \label{sect:diff}
As mentioned above, Lemmas \ref{lma:normalform}--\ref{lma:othertrans} are valid for $x<0$ or $x>0$, but not both together. In particular, the normal form transformation of Lemma \ref{lma:nftransfn} can be computed for $-\mu$ and for $\mu$ independently. This yields two distinct transformations, two distinct pairs of equations, and two separate phase spaces $\cc^2_-, \cc^2_+$. See Figure \ref{fig:alltrans}. Recall that our goal is a heteroclinic gluing argument; we wish to intersect an unstable manifold in $\cc^2_-$ with a stable manifold in $\cc^2_+$. However, since these manifolds do not lie in the same phase spaces, such an intersection is meaningless without some further justification. 
The next lemma provides the appropriate coordinate transformations to ``move" the unstable manifold from the phase space $\cc^2_-$ of the $y<0$ dynamics to the phase space $\cc^2_+$ of the $y>0$ dynamics. Then an intersection computed in the second phase space will meaningfully represent a transition from one invariant manifold to the second.

\begin{lma}\label{lma:movingtrans}
\change{For each $\ep>0$ sufficiently small, there exists a non-autonomous transformation $\mathcal T$ that maps the $y<0$ phase space to the $y>0$ phase space $\mathcal T \colon \cc^2_- \to \cc^2_+$. This transformation is defined by}
\begin{align}
\mathcal T(\ep) & \defeq 
\tau^{-1}\circ \nu^{-1} [\Id + \Phi_{1,1}(\ep^2)]^{-1}[\Id + \Phi_{1,1}(-\ep^2)]\circ \nu \circ \tau\\
& = \Id + \ep \begin{bmatrix} 
0 & 0 & 1 & 0 \\
\tfrac{-i}{2} & 0 & \frac{i}{2}e^{-2ix} & 0\\
0 & 0 & 0 & 0 \\
\tfrac{-i}{2}e^{2ix} & 0 & 0 & \tfrac{i}{2} \\
\end{bmatrix} + \mathcal{O}(\ep^2).
\end{align}
\end{lma}

\begin{proof} 
Figure \ref{fig:alltrans} shows how we arrive at the composition above. The form on the second line is obtained by direct computation.
\end{proof}

Now our task is to compute the unstable manifold of $(a,b) = (0,0)$ in the $y<0$ dynamics, transform it via $\mathcal T(\ep)$ into the phase space $\cc^2_+$ of the $y>0$ dynamics, and find its intersection with the strong stable foliation of periodic orbits. We first treat the $\ep = 0$ case.

\begin{figure}[h]
    \centering    
    $
        \xymatrix@C=.75in{
        {} & {u\in\rr} & {}\\
        {} & {U\in\rr^4}\ar[u] & {}\\
        {} & {(A,B)\in\cc^2}\ar[u]_{\Theta} & {}\\
        {} & {(A,B)\in\cc^2}\ar[u]_{\Phi_{3,0}} & {}\\
        {(A,B)\in\cc^2}\ar[ur]^{\Phi_{1,1}(-\ep^2)} & {} & {(A,B)\in\cc^2}\ar[ul]_{\Phi_{1,1}(\ep^2)}\\
        {(A_1,B_1)\in\cc^2}\ar[u]_{\nu} & {} & {(A_1,B_1)\in\cc^2}\ar[u]_{\nu}\\
        {(a,b)\in\cc^2_-}\ar[u]_{\tau} \ar[rr]^{\mathcal T(\ep)} & {} & {(a,b)\in\cc^2_+}\ar[u]_{\tau}\\
    }
    $
    \caption{A schematic summary of  transformations and variables.}
    \label{fig:alltrans}
\end{figure}
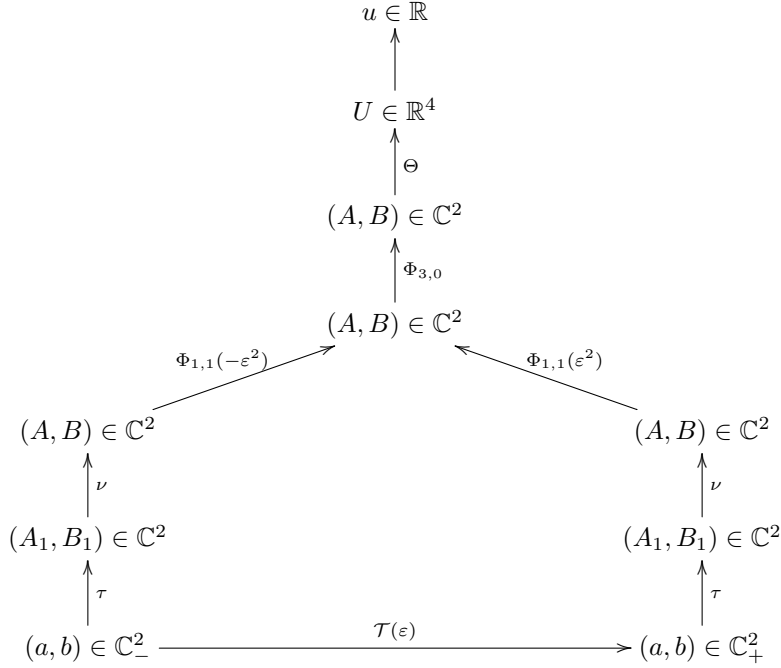

\subsection{The $\ep = 0$ Intersection} \label{sect:ep0}
We now describe the key structures in the dynamics of equations \eqref{eqn:firstord1}--\eqref{eqn:firstord2} with $\ep = 0$,\footnote{Note that for $y>0$, we have a version of the steady state real Ginzburg-Landau equations.}
\begin{align}
\label{eqn:1}
\frac{da}{dy} & = b\\
\label{eqn:2}
\frac{db}{dy} & =  
-\sgn(y) a+ a|a|^2.
\end{align}
The $\ep\to 0$ limit is justifiable by standard ODE theory, specifically the smoothness in parameters of invariant manifolds.

In the $y<0$ dynamics, the origin $(0,0)$ is a hyperbolic equilibrium with an unstable manifold $\mathcal{W}_0^u$. One may restrict to the real subspace, employ the conserved quantity $H^-(0)$, and apply the Gauge symmetry $\mathbf R_\phi$ to obtain an explicit parametrization\begin{align*}
\mathcal{W}_0^u & = \left\{ (e^{i\phi}r,e^{i\phi}\sqrt{r^2+\tfrac{1}{2}r^4})\mid r\in[0,\infty) \text{ and } \phi\in[0,2\pi)\right\}. 
\end{align*}

In the $y>0$ dynamics, there exists a family of periodic orbits
$$ (a_*,b_*)(y;\kappa) = (\sqrt{1-\kappa^2} e^{i\kappa y}, i\kappa \sqrt{1-\kappa^2} e^{i\kappa y}), \quad \quad -1<\kappa<1.$$
The subset  $\mathcal P = \{(a_*,b_*) \mid \kappa^2<\tfrac{1}{3}\}$ forms a normally hyperbolic manifold, which can be seen by computing the Floquet exponents of the linearization at $(a_*,b_*)$. The manifold $\mathcal P$ has a strong stable foliation $\mathcal W_{\mathcal P}^s$ which consists of a bounded and unbounded branch. For us, the bounded branch is relevant and may be constructed using the Gauge symmetry and the defect solutions
\begin{align} \label{eqn:Ad}
a_d(y;\kappa ) & = \left( \sqrt{2} \kappa  + i\sqrt{1-3\kappa ^2}\tanh{(\sqrt{1-3\kappa ^2}y/\sqrt{2})}\right)e^{i\kappa y},\quad\quad b_d(y;\kappa) =\pt{a_d}{y} , \quad \quad \text{for}\quad \quad \kappa^2<\tfrac{1}{3}.
\end{align}
Thus we also obtain an explicit parametrization for the bounded branch of the stable foliation,
\begin{align} \label{eqn:setstabfoln}
\mathcal W_{\mathcal{P}, b}^s = \left\{ \left(e^{i\psi} a_d(y;\kappa),e^{i\psi}b_d(y;\kappa)\right)\mid y\in \rr, \ \kappa^2<\tfrac{1}{3},\ \text{and}\ \psi\in[0,2\pi)\right\}\subset\mathcal W^s_\mathcal{P}.
\end{align}
Additional details of the $y>0$ dynamics may be found \cite[\S4]{scheelMorrissey} and references there.

Since $\mathcal T (0) = \Id$, we may immediately intersect the two invariant manifolds.

\begin{lma} \label{lma:intersection}
With $\ep = 0$, the intersection is given through 
\begin{equation} \label{eqn:0intersection}
\mathcal{W}_0^u \cap \mathcal{W}^s_\mathcal{P} = 
\left\{
  \left(
    e^{i\phi}\frac{1}{2},e^{i\phi}\frac{3}{4\sqrt 2}
  \right),
  \phi\in[0,2\pi)
\right\}
\end{equation}
and occurs for parameter values $r = \tfrac{1}{2}, y = \sqrt{2}\arctan\tfrac{1}{2}, \kappa = 0,$ and any $\phi,\psi$ such that $\psi = \phi-\frac{\pi}{2} \mod 2\pi$.
\end{lma}

\begin{proof}
We search only for the intersection $\mathcal{W}_0^u\cap \mathcal W_{\mathcal{P},b}^s$ with the bounded part of the stable manifold; one can verify that the unbounded branch of $\mathcal W_\mathcal{P}^s$ does not intersect the unstable manifold $\mathcal W_0^u$ but we omit the tedious details here since they are not relevant for our main result. We set equal the two parameterizations above and solve.

By using the conserved quantity $M =\Im( a \overline b)$, we can show that $\kappa =0$. This significantly simplifies the algebra and allows us to obtain an exact solution. First, note that $M\vert_{\mathcal{W}_0^u} = 0$, so $M\vert_{\mathcal W_{\mathcal{P},b}^s} = \Im(a_d\overline b_d) =0$. By the specific structure of $W_{\mathcal{P},b}^s$, we know that $\kappa$ is constant along solutions contained in the stable foliation (this is also evident from our parameterization). Since $\Im (a_*\overline{b_*}) = -\kappa(1-\kappa^2)$ and $(a_d,b_d)(\kappa) \to (a_*,b_*)(\kappa)$, we must have $-\kappa(1-\kappa^2) = 0$. Since $\kappa^2<\tfrac{1}{3}$, we see that $\kappa =0$. Plugging $\kappa =0$ into the equations and considering real and imaginary parts allows us to explicitly solve for the resulting parameter values.
\end{proof}

\subsection{The $\ep > 0$ Intersection}
Here we consider the dynamics of the truncated normal form equations \eqref{eqn:firstord1}--\eqref{eqn:firstord2} with $\ep>0$. We first establish that the intersection found in the $\ep = 0$ case is transverse, and thus persists for $\ep>0$. By standard ODE theory, the relevant structures $\mathcal W^u_{0}, \mathcal P,$ and  $\mathcal W_\mathcal{P}^s$ persist with an order-$\ep$ correction. We write $\mathcal W^u_{0,\ep}, \mathcal P(\ep),$ and $\mathcal W^s_{\mathcal P(\ep)}$ for these $\ep$-dependent objects.

\begin{lma} \label{lma:transverse}
There exists a neighborhood $\mathcal M\subset \rr_{\geq 0}$ of $0$ such that for $\ep\in\mathcal M$, the intersection
in Lemma \ref{lma:intersection} persists with an order-$\ep$ correction.
\end{lma}

\begin{proof}
We may write the same parameterizations as in Section \ref{sect:ep0} for $\mathcal W^u_{0,\ep}$ and $\mathcal W^s_{\mathcal P(\ep)}$ with placeholder terms of order $\mathcal O(\ep)$ (the specific nature of these terms is not relevant to the current argument).
After ``moving" $\mathcal W^u_{0,\ep}$ with the transformation from Lemma \ref{lma:movingtrans}, we have that 
\begin{align}\label{eqn:movedunstable}
 \widehat{\mathcal W^u_{0,\ep}} & \defeq \mathcal{T}(\ep)\left[\mathcal W^u_{0,\ep}\right]=\left\{ e^{i\phi}\left(r+\mathcal{O}(\ep^2),\sqrt{r^2+\tfrac{1}{2}r^4}+\ep(-ir/2+ir e^{-2i\phi}/2 )+\mathcal{O}(\ep^2)\right)\right\}.
\end{align}

Define $F(r,\phi,y,\kappa;\psi,\ep)$ as the difference of the parameterizations for the two manifolds. Thus, we have $F(\tfrac{1}{2},\phi,\sqrt{2}\arctan\tfrac{1}{2},0;\phi-\tfrac{\pi}{2},0)=0$ at the intersection. After a calculation, we find that for any $\phi\in [0,2\pi)$ we have 
\[
\det \left (D_{(r,\phi,y,\kappa)}F (\tfrac{1}{2},\phi,\sqrt{2}\arctan\tfrac{1}{2},0;\phi-\tfrac{\pi}{2},0) \right ) = -4
\]
So the derivative of $F$ is invertible at the intersection and the Implicit Function Theorem implies a nearby zero of $F$ for $\ep>0$ in a neighborhood $\mathcal M$.
\end{proof}

Our goal is to use an explicit expression for this intersection to obtain the order-$\ep$ correction to the wavenumber~$\kappa$. To do so, we may again use the parametrization for the the unstable manifold $\mathcal W^u_{0,\ep}$. However, we will need to use the conserved quantities in \eqref{eqn:1stintl} to easily access next order terms of $\mathcal W^s_{\mathcal P(\ep)}$. The full argument follows.


First, working in the $y<0$ phase space $\cc^2_-$, we will show that $\mathcal W^u_{0,\ep} = \mathcal W_0^u$ at order $\mathcal O (\ep)$. To see this, recall the conserved quantities \typo{$M$} and $H^-(\ep)$ expressed in equation \eqref{eqn:1stintl}. We compute that $\{ \nabla M, \nabla H^-(0)\}$ are linearly independent at the $\ep = 0$ intersection $\mathcal{W}_{0}^u\cap \mathcal W_\mathcal{P}^s$ given by equation \eqref{eqn:0intersection}. Thus the level set $\{M\equiv 0,\quad H^-(\ep)\equiv 0\}$ is a manifold with an explicit local approximation. 
Since this level set contains the origin, it contains the unstable manifold $\mathcal W^u_{0,\ep}$ (thus providing us with an approximate expression for $\mathcal W^u_{0,\ep}$ near the $\ep=0$ intersection). Notice that the leading-order $\ep$ of the conserved quantity $H^-(\ep)$ appears in a term with a factor  \typo{$M$}. Since $M\equiv 0$ on $\mathcal W^u_{0,\ep}$, this order-$\ep$ term has no effect.

Now we turn to the dynamics on the positive real line $y>0$. As mentioned above, we have the $\ep$-dependent family of periodic orbits $\mathcal{P}(\ep)$. Since \eqref{eqn:firstord1}--\eqref{eqn:firstord2} are equivariant under the Gauge symmetry $\mathbf R_\phi$, we may assume that the periodic orbits have the form 
\begin{align*}
a_*(\kappa ;\ep) & = s(\kappa ;\ep) e^{i \kappa y}\\
b_*(\kappa ;\ep) & = \left( p(\kappa ;\ep) + i q(\kappa ;\ep)\right) e^{i \kappa y}
\end{align*}
for real functions $s,p,q\colon ( -\tfrac{1}{\sqrt{3}},\tfrac{1}{\sqrt{3}}) \to \rr$. It is sufficient to restrict to a real amplitude for $a_*$ because we may combine the Gauge symmetry with the translation invariance of \eqref{eqn:firstord1}--\eqref{eqn:firstord2} to ``rotate" the solution pair so that $a_*$ is real when $y=0$. Plugging $(a_*,b_*)$ into \eqref{eqn:firstord1}--\eqref{eqn:firstord2}, we find that
\begin{align*}
s(\kappa;\ep) & = \sqrt{1-\kappa^2}-\ep \frac{\kappa^3}{4\sqrt{1-\kappa^2}} + \mathcal{O}(\ep^2)\\
p(\kappa;\ep) & = 0 + \mathcal{O}(\ep^2)\\
q(\kappa;\ep) & = k\sqrt{1-\kappa^2}-\ep \frac{-5\kappa^4+4\kappa^2-1}{8\sqrt{1-\kappa^2}} + \mathcal{O}(\ep^2).
\end{align*}

Next, we must investigate the $\mathcal O(\ep)$ terms of $\mathcal W_{\mathcal{P}(\ep)}^s$. As above, a direct computation shows that $\{ \nabla M, \nabla H^+(0)\}$ are linearly independent at the $\ep = 0$ intersection $\widehat{\mathcal{W}_{0}^u}\cap \mathcal W_{\mathcal{P}}^s$. Thus, the level set of the conserved quantities forms a manifold with a local expression near the intersection $\widehat{\mathcal{W}_{0,\ep}^u}\cap \mathcal W_{\mathcal{P}(\ep)}^s$ for small $\ep>0$.

Evaluating each of the conserved quantities at the relative equilibria $(a_*,b_*)$, we obtain functions of $\kappa,\ep$
\begin{align*} 
\mathcal{M} (\kappa;\ep) & \defeq M\vert_{(a_*,b_*)} =  -\kappa(1-\kappa^2) + \ep \frac{1}{8}(7\kappa^4-4\kappa^2+1)+\mathcal{O}(\ep^2)\\
\mathcal{H}(\kappa;\ep) & \defeq H^+(\ep)\vert_{(a_*,b_*)} = \frac{1}{2}(1-\kappa^2)(1+3\kappa^2) + \ep \frac{1}{2} \kappa(-4\kappa^4+3\kappa^2-1)+\mathcal{O}(\ep^2).
\end{align*}
For each $\kappa^2<\tfrac{1}{3}$, the level set $\{M\equiv \mathcal{M}(\kappa;\ep) ,\quad H^+(\ep)\equiv \mathcal{H}(\kappa;\ep) \}$ is a manifold containing the periodic solution $(a_*,b_*)(\kappa;\ep)$ and thus contains the stable manifold of the periodic solution. Thus, we have an expression (valid near the intersection $\widehat{\mathcal{W}_{0,\ep}^u}\cap \mathcal W_{\mathcal{P}(\ep)}^s$) which gives a leading-order approximation of the strong stable manifold of the periodic orbits
$$\mathcal{W}_{\mathcal{P}(\ep)}^s \subseteq \bigcup_{\kappa^2<\tfrac{1}{3}}\{ M =\mathcal{M} (\kappa;\ep), \quad H^+(\ep) = \mathcal{H} (\kappa;\ep) \}.$$

Evaluating \typo{$M, H^+(\ep)$} at the ``moved" unstable manifold $\widehat{\mathcal W^u_{0,\ep}}$ gives us
$$ M\vert_{\widehat{\mathcal W^u_{0,\ep}}} = \ep \frac{r^2}{2}(1-\cos{2\phi}), \quad \quad H^+(\ep)\vert_{\widehat{\mathcal W^u_{0,\ep}}} = 2r^2 +\ep r \sqrt{r^2+\tfrac{1}{2}r^4}\sin{2\phi} + \ep^2(1-\cos {2\phi})(r^2+\frac{r^4}{8}).$$
As above, these now provide leading-order approximations for the level set in the $y>0$ phase space which contains the transformed unstable manifold $\widehat{\mathcal W^u_{0,\ep}}$.

We are now ready to prove the following lemma.

\begin{lma}
For $\ep>0$, there is a one-dimensional intersection $\widehat{\mathcal W_{0,\ep}^u}\cap \mathcal{W}_{\mathcal{P}(\ep)}^s$. On the intersection, one has
\begin{align*}
\kappa & = \frac{\ep}{8}\cos(2\phi) + \mathcal O(\ep^2).
\end{align*}
\end{lma}

\begin{proof}
The first statement is a consequence of Lemma \ref{lma:transverse}. For the second, we set 
$$ M\vert_{\widehat{\mathcal W^u_{0,\ep}}} = \mathcal{M}(\kappa;\ep), \quad \quad \text{and}\quad \quad H^+(\ep)\vert_{\widehat{\mathcal W^u_{0,\ep}}} = \mathcal{H}(\kappa;\ep), $$
expand in $\ep$, and compare terms of the same order in $\ep$. Solving the resulting system of equations, we obtain
$\kappa = \frac{\ep}{8}\cos{(2\phi)}$ and $r = \tfrac{1}{2}$ as one solution.

When $\ep=0$, we know the intersection occurs at $\kappa=0$ and with $r=\tfrac{1}{2}$. By the transversality mentioned above, we know that the intersection persists and remains unique under the $\ep$-perturbation. Other solution pairs $(\kappa,r)$ do not have the property that $(\kappa,r)\to (0,\tfrac{1}{2})$ as $\ep\to 0$.
\end{proof}

\subsection{Proof of Theorem \ref{thm:mainthm}}

\begin{proof}[Proof of Theorem \ref{thm:mainthm}]
Let $\mathcal{V}\subseteq \rr_{\geq 0}$ be a neighborhood of $0$ contained in the neighborhoods guaranteed by Lemmas \ref{lma:normalform} and \ref{lma:transverse}. Take $\mu \in \mathcal{V}$ and for each $\phi\in[0,2\pi)$, let $(a,b)_\phi\in\cc^2$ be a point on the intersection $\widehat{\mathcal W_{0,\ep}^u}\cap \mathcal{W}_{\mathcal{P}(\ep)}^s$ such that $\phi$ is the phase of the leading order terms of $a$ as in the parameterization in \eqref{eqn:movedunstable}. The existence of $(a,b)_\phi$ is guaranteed by the Implicit Function Theorem in Lemma \ref{lma:transverse}. Let $\Psi(\mu) = \Theta \circ (\Id +\Phi_{3,0})\circ (\Id + \Phi_{1,1}(\mu))$ be the full normal form transformation from Section \ref{subsect:NF}. Let $U(x;\theta)$ be a solution to equation \eqref{eqn:U} with initial condition 
$$U(0;\theta) = \Psi(\mu)\left(\nu\circ\tau((a,b)_\phi)); \mu \right).$$
Take $u(x;\theta) = U_1(x;\theta)$, the first component. 

In the limit $x\to -\infty$, since 
$$[\Psi(-\mu)]^{-1}\left(\nu^{-1}\circ\tau^{-1} (U(x;\theta))\right)\in \mathcal W_{0,\ep}^u,$$
we know that $U(x;\theta)$ is on the unstable manifold of $0\in\rr^4$. Thus $|u(x;\theta)|\to 0$ as $x\to -\infty$. 

Next, consider the behavior of $u$ as $x\to \infty$. Note that $U(x;\theta)$ is on a strong stable fiber of a periodic orbit
$$U^*(x;\theta) = \Psi\left(\nu\circ\tau((a_*,b_*)(y;\kappa);\mu\right) \quad \quad \text{where} \quad \quad \kappa = \frac{\sqrt \mu}{8}\cos{(2\phi)} + \mathcal{O}(\ep^2).$$
After applying the various transformations, we see that the first component is $U^*_1(x;\theta) = u_*(kx-\theta;k)$, a periodic function with wavenumber 
\begin{equation}\label{e:kk}
k(\phi;\mu) = 1 + \frac{\mu}{16}\cos{(2\phi)} + \mathcal{O}(\mu^{2}). 
 \end{equation}
As $x\to \infty$ we see that $|u(x;\theta) - u_*(k x-\theta;k)| = |U_1(x;\theta) - U^*_1(x;\theta)| \to 0$. 

All that's left is to establish the relation between the phase $\phi$ at the intersection and the phase $\theta$ of the asymptotic pattern. By comparing the phase of the defect solutions $a_d(y;\kappa)$ at the intersection $\widehat{\mathcal W_{0,\ep}^s}\cap \mathcal{W}_{\mathcal{P}(\ep)}^s$ and in the limit $y\to \infty$, we can establish the relation
$$ \theta(\phi;\ep) = \phi -\frac{\ep}{4\sqrt{2}}\cos(2\phi) + \mathcal{O}(\ep^2). $$
Substituting this into \eqref{e:kk}, we obtain the same leading order expansion now with $\phi$ replaced by $\theta$
\begin{align*}
k(\theta;\mu) & = 1 + \frac{\mu}{16}\cos\left(2\left(\theta + \mathcal O(\sqrt{\mu})\right)\right) + \mathcal{O}(\mu^{2})\\
& = 1 + \frac{\mu}{16}\cos{(2\theta)} + \mathcal{O}(\mu^{3/2}).
\end{align*}
\end{proof}


\begin{figure}[h]
    \centering
    \includegraphics[width = 0.4 \textwidth]{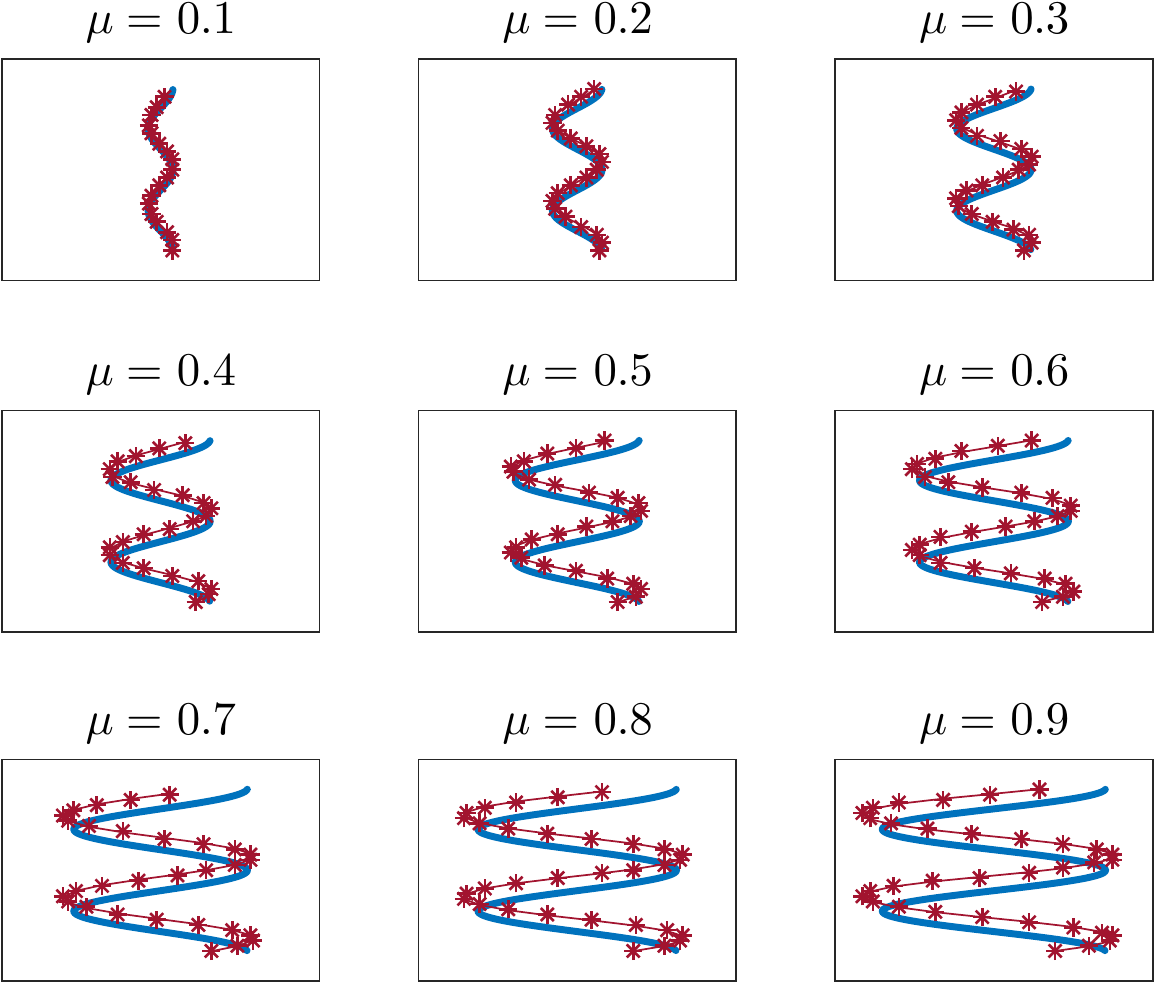}
    \qquad
    \includegraphics[width = 0.4 \textwidth]{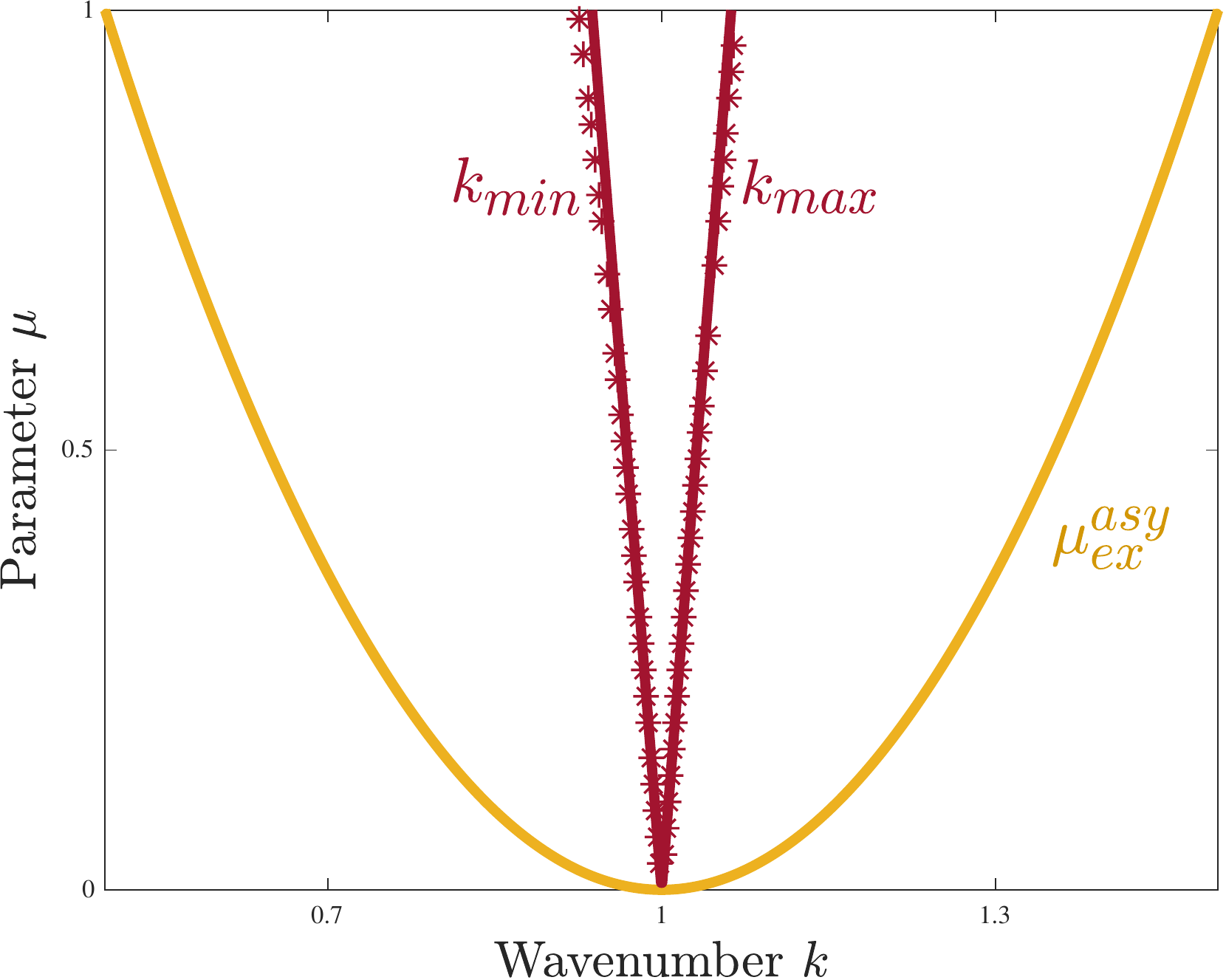}
    \caption{
    Left: Strain-Displacement relations computed with numerical continuation (red $\ast$) and predicted by theory (blue). \change{Each curve is plotted with $k\in [0.92,1.08]$, horizontal, and $\phi\in [-4.25,4.25]$, vertical.} Right: Maximal/minimal wavenumbers computed (red $\ast$) and predicted with/out (gold/red) jump-type parameter.}\label{fig:AsymComp}
%
\end{figure}

\section{Numerics and Extension} \label{sect:application}

\subsection{Numerical Corroboration} 
We found excellent agreement between our theory and values of $k,\theta$ computed using numerical continuation. Numerical computation of strain-displacement relations is equivalent to computing a family of heteroclinic orbits connecting an equilibrium and a family of periodic orbits. Our particular approach, following \cite[\S5]{scheelMorrissey}, uses numerical farfield-core decomposition. More specifically, we use an ansatz $u(x)=\chi(x) u_*(kx-\theta;k)+w(x)$ with $\chi$ a smoothed out version of the characteristic function of $[1,\infty)$, say. One solves for the correction $w$ and the parameters $\theta,k$ after adding artificial homogeneous Dirichlet boundary conditions at $\pm L$ and imposing a phase condition near $x=L$ to enforce exponential localization of $w$. 
For details, see \cite[\S5]{scheelMorrissey}. The results indicate strong agreement with Theorem \ref{thm:mainthm}, even as $\mu$ approaches 1. See Figure \ref{fig:AsymComp}.

In the remainder of this section we consider implications in two space-dimensions, considering \eqref{eqn:SHjump} with $(x,y)\in\rr^2$. First, we review the known two-dimensional instabilities of the equation with \typo{a} homogeneous parameter. We then discuss how slowly growing the size of the domain can select a \emph{unique} wavenumber from the band of allowed wavenumbers \cite{gohBeekie}. Finally, we present results of direct numerical simulations and summarize our methods.

\subsection{Instabilities}

Considering our one-dimensional stripe solutions  $u_*(kx;k)$ as solutions of the two-dimensional Swift-Hohenberg equation, posed on $(x,y)\in\rr^2$ \change{with a homogeneous parameter}, we may study the spectrum of the linearization at one of these states. It turns out that these stripes are linearly (and nonlinearly) stable only in a small subregion, bounded by curves that are commonly referred to as Eckhaus and zigzag boundaries, which possess asymptotic expansions near $\mu=0$,
\[
 \mu_E=3(1-k^2)^2+\mathcal O ((1-k^2)^3),\qquad (1-k^2)=-\mu_Z^2/512+\mathcal O (\mu_Z^3),
\]
respectively. The Eckhaus boundary is present in one-dimensional systems, whereas the zigzag boundary invokes perturbations depending on $y$. See \cite{mielke} and references therein for a detailed account of stability. \change{We illustrate these stability boundaries in Figure \ref{fig:stepParam}.}

\begin{figure}[h]
    \centering
    \includegraphics[width = 0.9 \linewidth]{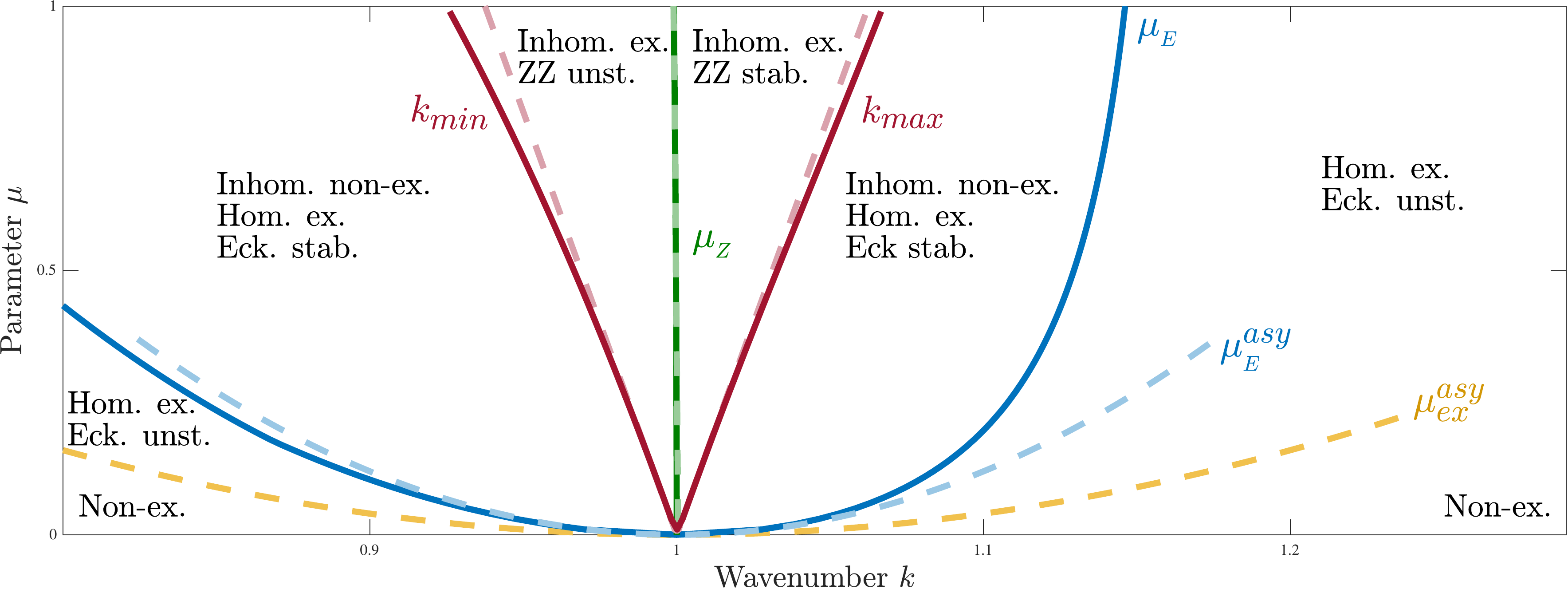}
    \caption{\change{Regions of in/stability and existence for stripes in a system with homogeneous parameter and for half-stripes with a jump-type inhomogeneity; leading order asymptotics (dashed) and directly computed data (solid). See text for details.}}\label{fig:stepParam}
\end{figure}
\change{
After overlaying the regions of stability with the existence information (for both full stripes and half-stripes) contained in Figure \ref{fig:AsymComp}, we have eight distinct regions. These appear in Figure \ref{fig:stepParam} and are described in detail below, clockwise from bottom left:}
\begin{itemize}
\item no stripes exist\\ (Non-ex.)
\item stripes exist with homogeneous parameter and are Eckhaus unstable, no half-stripes exist\\ (Hom. ex., Eck. unst.)
\item no half-stripes with inhomogeneity, full stripes exist and are Eckhaus stable but zigzag unstable\\ (Inhom. non-ex., Hom. ex., Eck. stab.)
\item half-stripes exist with inhomogeneity and are zigzag unstable, full stripes exist with previous stability\\ (Inhom. ex., ZZ unst.)
\item half-stripes exist and are stable, full stripes exist and are stable\\ (Inhom. ex., ZZ stab.)
\item no half-stripes exist, full stripes exist and are stable\\ (Inhom. non-ex., Hom. ex., Eck. stab.)
\item stripes exist and are Eckhaus unstable, no half-stripes exist\\ (Hom. ex., Eck. unst.)
\item no stripes exist\\ (Non-ex.)
\end{itemize}
\change{Note that half-stripes selected by the jump-type inhomogeneity have wavenumbers that are bounded away from the Eckhaus instability, but overlap significantly with the zigzag-unstable region. }

\subsection{Growing Domains and Zigzag Selection}
\change{The fact that the minimum selected wavenumber $k_{min}$ is zigzag unstable has dramatic consequences for pattern selection. However, the effects are not immediately visible in simulations because the system does not select a \emph{unique} wavenumber. In particular, we would need to impose a phase condition $\theta \sim \pi/2$ to select $k_{min}$. This would represent an additional, external mechanism for selection. Instead, we pose our problem in the context of slowly growing domains and apply theory from \cite{gohBeekie}. In the second part of this section, we illustrate the results with direct numerical simulations that indicate that the system self-selects a zigzag pattern.}

Systems with slowly growing domains have been observed to exhibit pattern selection, e.g. in developing organisms in biology \cite{kochMeinhardt}. In \cite{gohBeekie}, several model equations similar to ours are examined with a moving parameter jump $m(x-ct)$ or on a linearly growing domain with boundary conditions. The main result is that the selected wavenumber $k$ depends on the speed $c$ according to
$$k(c) = k_{min} +\mathcal{O}(\sqrt c)$$
where $k_{min}$ is the minimum wavenumber in the admissible band of wavenumbers occurring in the same system with a fixed  parameter jump. In our problem \eqref{eqn:SHjump}, we expect $k_{min}$ to be zigzag unstable, see Figure \ref{fig:stepParam}. Thus, if we choose $c$ near $0$, we may expect the system to select striped pattern solutions which are unstable to transverse perturbations.

\subsubsection*{Simulations}

We observed the selection and development of zigzag patterns in a direct numerical simulation of the problem \eqref{eqn:SHjump} in two dimensions. We use a standard spectral method with $2^{11} \times 2^{11}$ Fourier modes, creating an effective $dx\approx 0.01$, and implicit Euler time stepping with $dt = 0.025$.
\begin{figure}[h]
\centering
$\begin{array}{c c c}
\includegraphics[width = .3\textwidth]{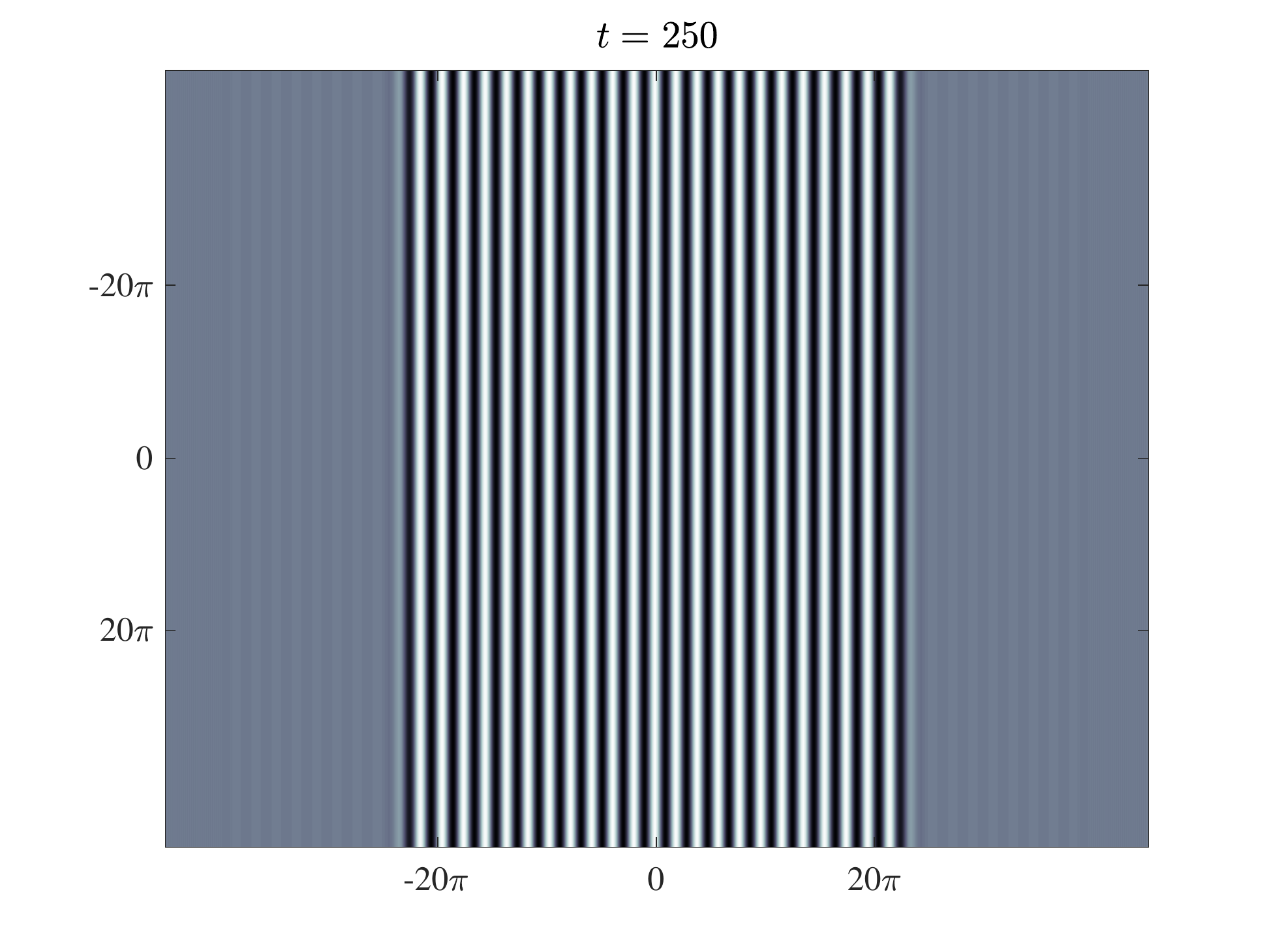} & \includegraphics[width = .3\textwidth]{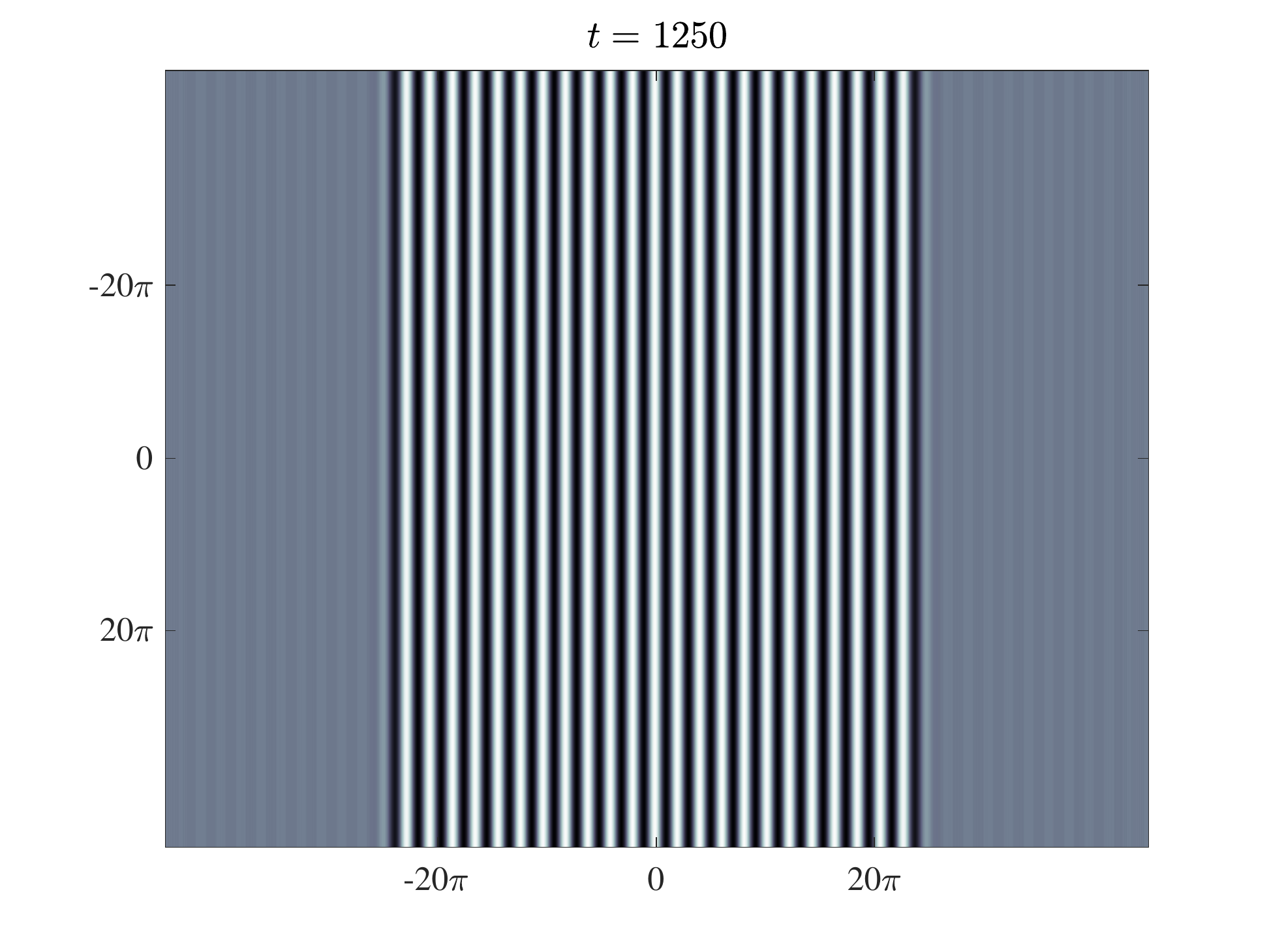} & \includegraphics[width = .3\textwidth]{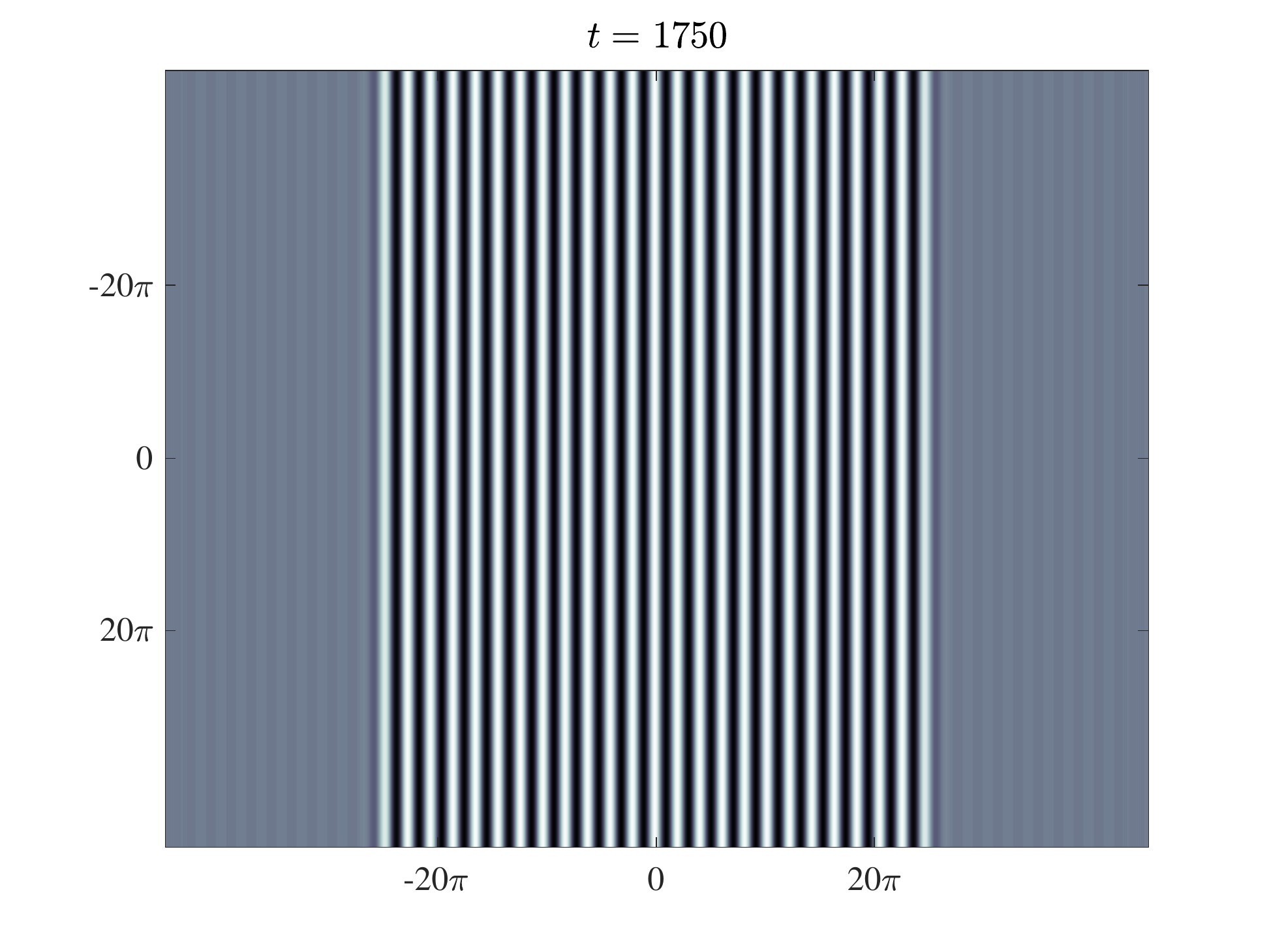}\\ 
\includegraphics[width = .3\textwidth]{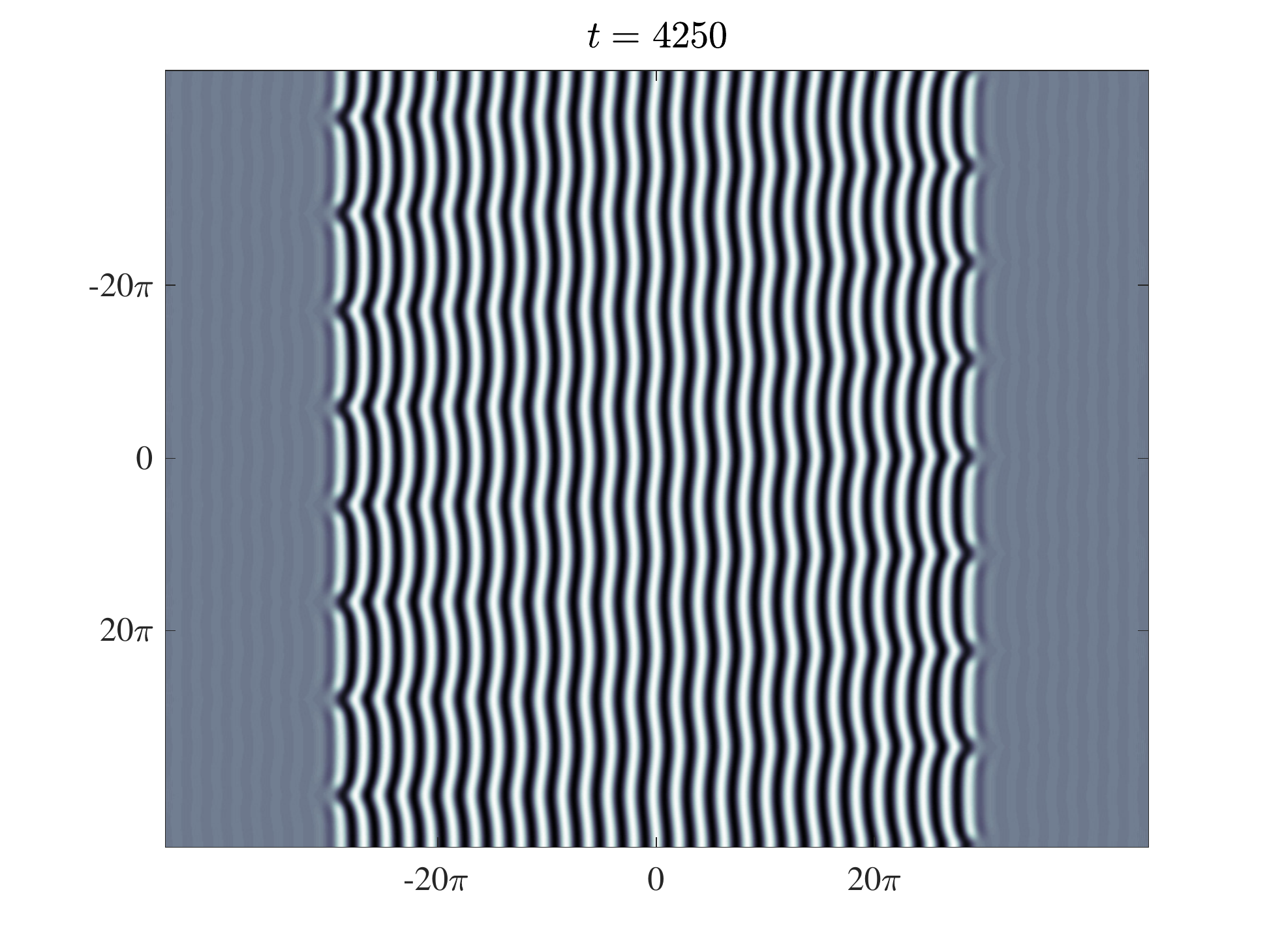} & \includegraphics[width = .3\textwidth]{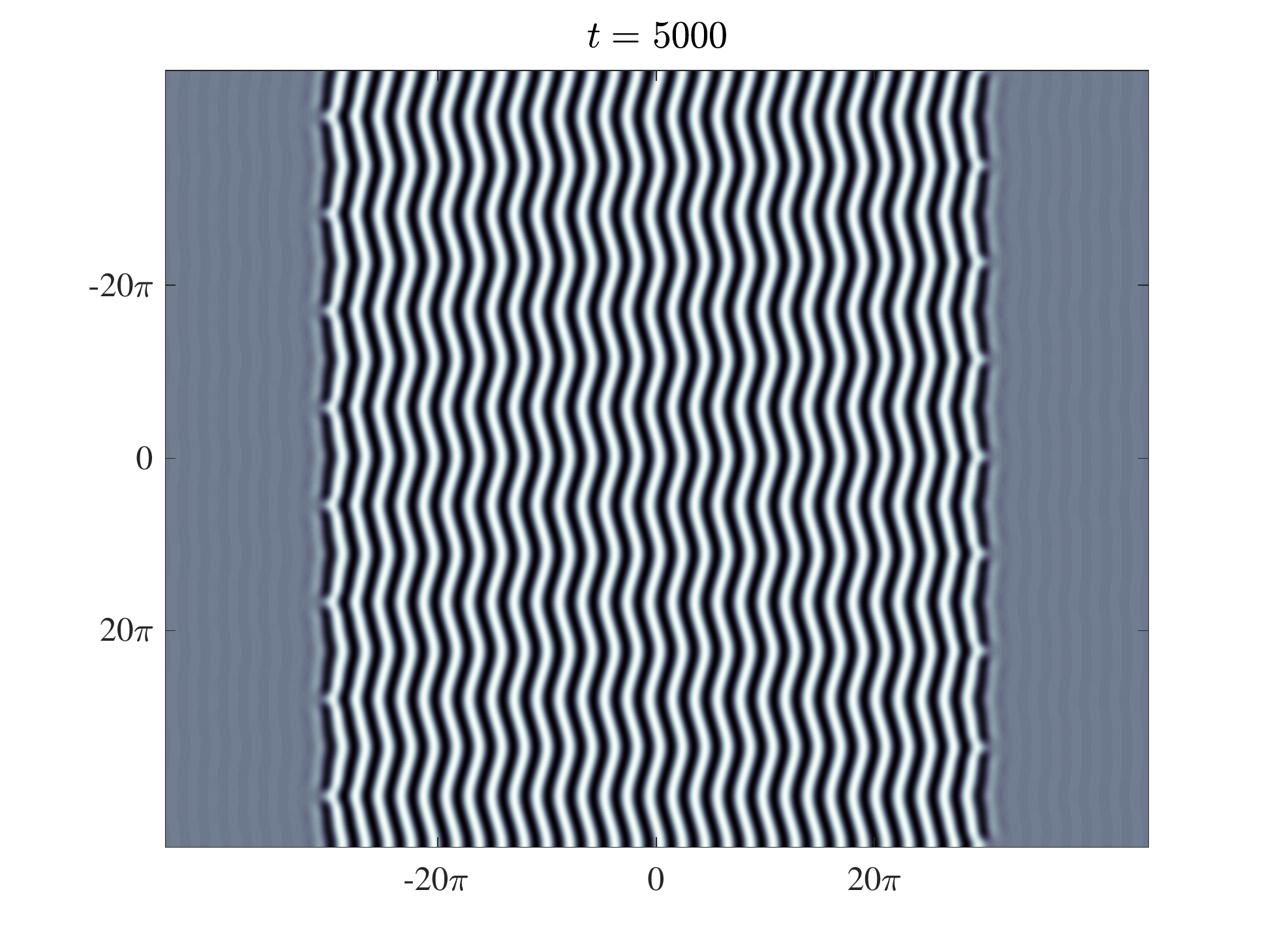} & \includegraphics[width = .3\textwidth]{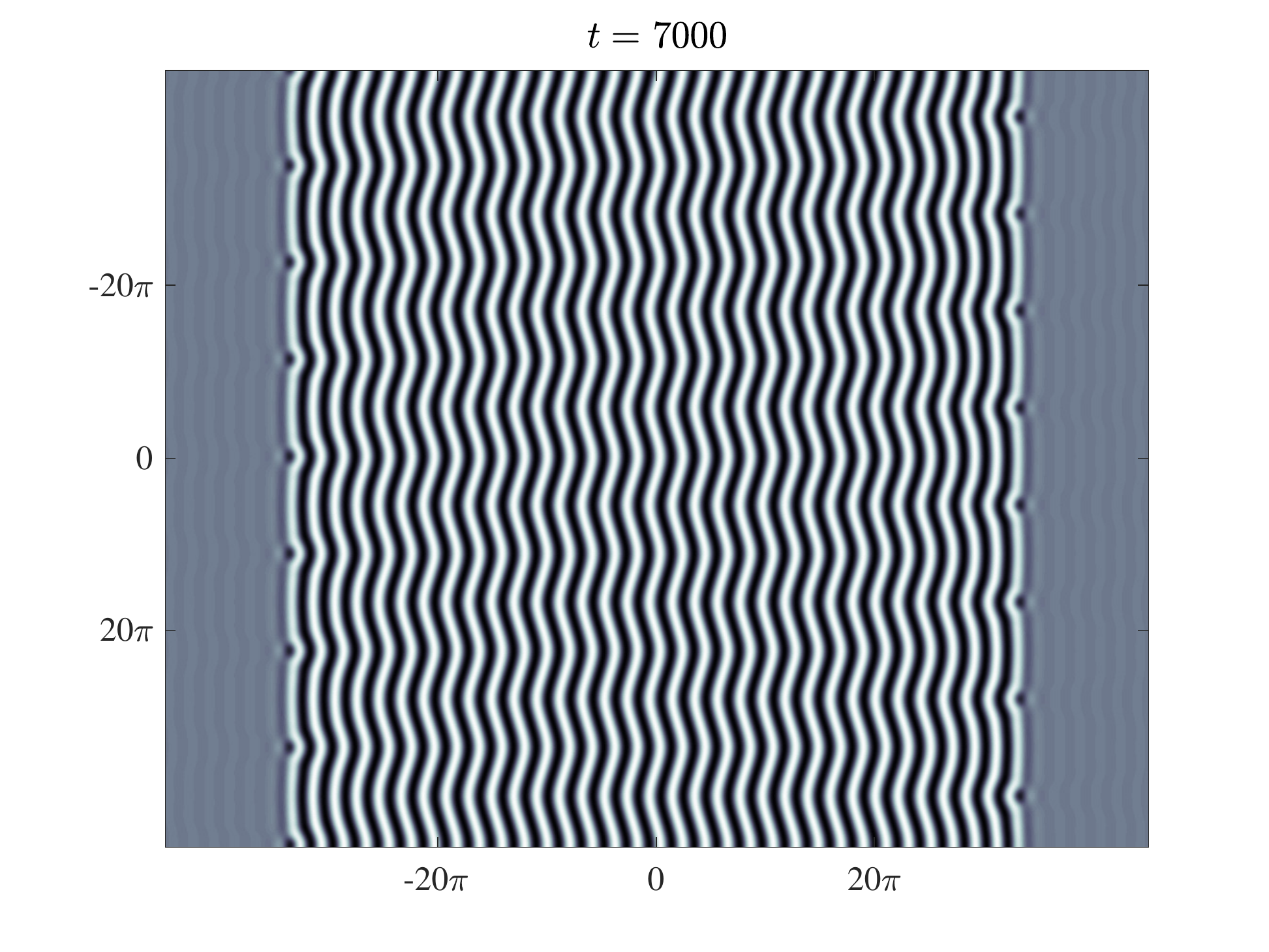}\\ 
\end{array}$
\caption{Snapshots from direct numerical simulation of the Swift-Hohenberg equation posed on a square domain with periodic boundary conditions; see text for detailed description. The simulation uses a spectral method with implicit Euler time-stepping and initial data that is periodic in the middle half of the domain and $0$ elsewhere.}
\label{fig:frames}
\end{figure} 
All simulations are posed on a square domain with side length $2 L = 90\pi$ and periodic boundary conditions. \change{We use a ``plateau"-type parameter consisting of two jumps for compatibility with our periodic boundary conditions, which are necessary due to our use of a spectral method. (In this case, the effect from the additional jump in negligible, see Section \ref{sect:discussion}.) The two jumps move away from each other to emulate a growing domain. The size of both jumps is $\mu = 0.8$ and initial data is an even periodic pattern with initial wavenumber $k_0$ in the middle half of the domain}
$$u(x,y,0) = \begin{cases} \sqrt\mu \cos (k_0 x), & |x|<L/2\\ 0, & |x|>L/2\end{cases}.$$
We fix $c = 0.005$ and $k_0 = 1.05$. At $t_p=1500$ we add a small, transverse perturbation 
$$ u_p(x,y,t_p) = \begin{cases} 0.1 \cos \big(k_p x + 0.9 \cos(8\pi y/L)\big), & |x|<L/2+c t_p\\ 0, & |x|>L/2+c t_p\end{cases}$$
where $k_p = 0.97125$ is the observable wavenumber achieved by the system at $t_p$ as numerically computed in independent trials.

Figure \ref{fig:frames} shows snapshots for various $t$ values. When $t<1500$, all stripes widen as the wavenumber decreases and the parameter plateau widens; no new stripes appear, should the reader choose to count. At $t=1750$, a stripe has been added on the outside, indicating that the wavenumber has stopped decreasing. Also, notice that the transverse perturbation, added at $t=1500$, is small enough that it is unobservable. At $t= 4250$ we can see the system moving away from the $y$-invariant stripes. By $t = 5000$, the system seems to have relaxed to a stationary zigzag pattern. Indeed, for long times past $t=7000$ the zigzag pattern remains stationary and continues to add new zigs and zags on the outside. 
We comment further on patterns in bounded regions in the next section.

\section{Discussion} \label{sect:discussion}

\subsubsection*{Patterns in Bounded Regions}

Experiments are usually performed in bounded domains, and one may therefore be interested in how the wavenumber selection mechanism described here interacts with left and right boundary conditions. \change{In the center of the pattern-forming region, one expects that a strain-displacement relation is induced by the boundary at one side while the second boundary forces a different strain-displacement relation.}
Matching these two relations, and including correction terms that are exponentially small in the size of the patterned region, has been carried out in \cite[\S5]{scheelMorrissey}. The result is a geometric subtraction (subtracting phases) and quantization (intersecting with $\phi=2 k L$) of the two strain-displacement relations. The argument there carries through in a straightforward fashion to the equivalent case of a double-jump, or parameter plateau, 
\[
 m(x)=\left\{\begin{array}{ll}\mu,& |x|\leq L,\\ -\mu,& |x|\leq -L,\end{array}\right.
\]
with $L$ large, confining the patterned region to a strip $|x|\leq L$. The set of equilibria can then be described ``explicitly'' in the thermodynamic limit of $L\to\infty$. Numerical results in a similar situation were presented in \cite[Fig. 3]{gohBeekie}.

\subsubsection*{Boundary Conditions}

There is an undeniable analogy to be made between our problem \eqref{eqn:SHjump} and the SH equation posed on a semi-infinite domain with Dirichlet boundary conditions
\begin{align} \label{eqn:SHBC}
\ut = -\left(1 + \frac{\partial^2}{\partial x^2}\right)^2 u +\mu u - u^3, \quad \quad u\in \rr, \quad\quad x\in[0,\infty) \quad \quad u(0,t) = \dstypo{u_{x}}(0,t) = 0.
\end{align} 
Problem \eqref{eqn:SHBC} was formally and numerically explored during the early 1980s. 
In particular, the linear coefficient $\frac{1}{16}$ from our Theorem \ref{thm:mainthm} also appears in \cite{CDHS, ferry, pomeauZaleskiBC}. The authors' methods use conserved quantities of amplitude equations introduced by \cite{newellWhitehead} which are equivalent to our real Ginzburg-Landau equations \eqref{eqn:1}-\eqref{eqn:2} with $y>0$. We have found no results from this period providing an explicit strain-displacement relation between the phase and wavenumber, although the idea is mentioned in \cite{pomeauZaleskiBC}. 

On the other hand, one can envision a homotopy from parameter jumps to boundary conditions, treating parameter jumps $m(x)=\mu$ for $x>0$, as before, but $m(x)=-\mu_-<0$ for $x<0$. 
Slightly generalizing our result to this scenario with $\mu_-=- C\mu$ for some $C>0$, one finds the same leading-order expansion with band width boundaries $1\pm \mu/16$ . Letting $C\to \infty$, or even using different scalings in $\mu$, one can arrive at Dirichlet boundary conditions in Ginzburg-Landau, or clamped boundary conditions in the SH equation, $u_x=u=0$ at $x=0.$ It would be interesting to study strain-displacement relations in such a broader class of parameter jumps, testing the universality of the $\mu/16$-correction.

More recently, a similar boundary value problem was studied numerically through the strain-displacement framework in \cite{scheelMorrissey}.
\typo{The main difference from our earlier analysis is that there is no unstable manifold of $u\equiv 0$ to consider. Instead, one simply intersect a boundary manifold, for instance $\mathcal B = \{ u \mid u+u_{xx}=u_x+u_{xxx} =0 \text{ at } x=0\}$, with the stable manifold of the periodic solutions $u_*$.} 


\subsubsection*{Slow Parameter Ramps}

Even with the above attention to the problem with boundary conditions, a rigorous discussion of the problem with a jump-type parameter is absent from the literature.
Instead, some authors have considered a parameter that varies slowly in space. The first such study appears in \cite{CrossKramer} which shows that at any finite order, a sufficiently slow spatial parameter ramp selects a unique wavenumber. Given a maximal value of the parameter ramp, the authors of \cite{pomeauZaleskiRamps} compute the selected wavenumber to leading order. This contrasts with our case, in which a narrow band of wavenumbers is selected. One consequence of this qualitative difference is that the selected wavenumber is zigzag stable, and thus a slow parameter ramp cannot be used to select zigzag patterns in the way that we use a parameter jump to do so in Section \ref{sect:application}.

It would be interesting to understand in more detail the transition from slow parameter ramps to boundary conditions or the parameter jump we consider here. One approach would be to interpolate between a slow-ramp and a jump-type parameter. One could consider heterogeneous parameter profiles $h(x;\gamma)$ which converge (in some sense) to a jump-type parameter and an arbitrarily slow ramp for extremal values of $\gamma$. A simple such family is given by $h(x,\gamma) = \mu \tanh(\gamma x)$.
We expect that for $\gamma$ sufficiently large, our analysis here can be adapted without much additional work, leaving leading-order coefficients unchanged. Roughly speaking, one appends an equation for the parameter evolution, $h'=(\gamma/\mu)(\mu-h)(\mu+h)$. Then the flows we found for $x>0$ and $x<0$ now reside in normally hyperbolic asymptotic subspaces where $h=\pm\mu$. The fast flow in the direction of $h$ is trivial in the direction of $u$. Thus contributions to the strain-displacement relation come from normal form transformations, only, as in the present work. We expect changes in the asymptotics when normal hyperbolicity breaks down, for $\gamma\sim \sqrt{\mu}$, such that dynamics in $h$ cannot be thought of as instantaneous anymore and normal-form coordinate changes evolve in $h$ nontrivially. In the limit when $\gamma$ is very small, normal form changes can be performed adiabatically in $h$ and we recover the results on slow ramps. 

\subsubsection*{Stability}

Linearizing at the half-stripe solutions constructed here, one expects to find continuous spectrum up to the origin. As argued in \cite{scheelMorrissey}, an eigenvalue emerges from the edge of the essential spectrum when following solutions along the strain-displacement relation through an extremum of $k$. As a consequence, the slope, $k'(\theta)$ gives a parity index for stability. Generally, one might expect that decreasing $k$, that is, stretching the asymptotic pattern, would be associated with ``pulling'' on the pattern, that is, displacing the pattern to the right or increasing $\theta$. With this intuition, $k'(\theta)<0$ would correspond to stable patterns and $k'(\theta)>0$ to ``unphysical'' unstable patterns; see the discussion in \cite[\S 2.4]{scheelMorrissey} and the numerical evidence in \cite{gohBeekie}. The present situation might be a good starting point to understand how this mechanistic intuition may relate to a spectral analysis of the linearization. 

\subsubsection*{Patterns in the Plane: Two Dimensions}

In two space dimensions $(x,y)\in\rr^2$, with a one-dimensional parameter jump $m(x)$,  stripes near $x=\infty$ can possess arbitrary orientations $u_*(k_x x +k_y y;k)$, $k^2=k_x^2+k_y^2$. In particular, one can now ask for solutions asymptotic to stripes that are perpendicular to the parameter jump, with $k_x=0$, or at an oblique angle to the parameter jump. In the much simpler Allen-Cahn equation, and to some extent in the slightly more complicated Cahn-Hilliard equation, such solutions have been constructed in \cite{monteiro1,monteiro2,monteiro3}, showing in particular that stripes are either parallel or perpendicular to the parameter jump in this case. In the case of the Swift-Hohenberg equation, pursuing an infinite-dimensional center-manifold and normal form analysis following the analysis of grain boundaries in \cite{hs1,hs2,sw} seems to be a promising generalization of the approach developed here. 

Yet more intricate phenomena are to be expected when considering hexagonal patterns rather than stripes, for instance when adding quadratic nonlinearities to the Swift-Hohenberg equation. In simple cases, such as when the parameter jump aligns with a symmetry axis of the hexagonal pattern, one could try to adapt the infinite-dimensional analysis for oblique stripes and track quadratic terms in the normal form \cite{doelmanSchneider}, possibly pointing towards distortions in of the perfect hexagonal lattice due to the parameter jump.

%

\end{document}